\documentclass[11pt]{amsart}
\usepackage{amsmath,amsfonts, amsthm, amssymb}
\usepackage{hyperref}
\usepackage{upgreek}
\usepackage{microtype}
\usepackage{graphicx}
\usepackage{tikz}
\usetikzlibrary{calc}
\usetikzlibrary{decorations.pathmorphing}
\usetikzlibrary{shapes.misc}

\tikzset{snake it/.style={decorate, decoration={snake, amplitude=0.5mm, segment length=3mm}}}
\tikzset{cross/.style={cross out, draw=black},
cross/.default={1pt}}

\usepackage[style=alphabetic, backend=biber,maxbibnames=99]{biblatex}
\def\Real{{\mathbb R}}
\def\R{{\mathbb{R}}}
\def\bbC{{\mathbb C}}

\def\bbH{{\mathbb H}}

\def\bbZ{{\mathbb Z}}
\def\Z{{\mathbb{Z}}}
\def\bbT{{\mathbb{T}}}
\def\Torus{{\mathbb T}}

\def\Expec{{\mathbb E}}

\def\eps{{\varepsilon}}
\def\One{{\mathbf{1}}}

\DeclareMathOperator{\Div}{div}
\newcommand{\diff}{\mathop{}\!\mathrm{d}}
\newcommand{\Id}{{\operatorname*{Id}}}
\newcommand{\Ft}[1]{{\widehat{#1}}}
\newcommand{\supp}{\operatorname*{supp}}

\newcommand{\Impt}{\operatorname*{Im}}
\newcommand{\Rept}{\operatorname*{Re}}

\def\lsim{{\,\lesssim\,}}

\def\bra#1{\mathinner{\langle{#1}|}}
\def\ket#1{\mathinner{|{#1}\rangle}}
\def\braket#1{\mathinner{\langle{#1}\rangle}}
\def\Braket#1{\langle{#1}\rangle}

\newcommand{\mcal}[1]{{\mathcal{#1}}}

\newcommand{\ldom}{\prec}

\numberwithin{equation}{section}

\newtheorem{theorem}{Theorem}

\numberwithin{theorem}{section}
\newtheorem{lemma}[theorem]{Lemma}

\newtheorem{corollary}[theorem]{Corollary}
\newtheorem{proposition}[theorem]{Proposition}

\title{Quantitative delocalization estimates for the tight-binding model}
\title{The scaling limit of the resolvent in the Anderson tight-binding model}
\title{Self-consistent equations and Quantum Diffusion for the Anderson model}

\author{Adam Black}
\address{Max Planck Institute for Mathematics in the Sciences, Leipzig, Germany}
\email{adam.black@mis.mpg.de}
\author{Reuben Drogin}
\address{Department of Mathematics, Yale University, New Haven, CT}
\email{reuben.drogin@yale.edu}
\author{Felipe Hern\'{a}ndez}
\address{Department of Mathematics, Penn State University, State College, PA}
\email{felipeh@psu.edu}

\date{\today}

\begin{document}
\begin{abstract}
We consider the Anderson tight-binding model on $\Z^d$, $d\geq 2$, with Gaussian noise and at low disorder $\lambda>0$. We derive a diffusive scaling limit for the entries of the resolvent $R(z)$ at imaginary part $\Impt z\sim\lambda^{2+\kappa_d}$, $\kappa_d>0$, with high probability.
As consequences, we establish quantum diffusion (in a time-averaged sense) for the Schr\"{o}dinger propagator at the longest timescale known to date and improve the best available lower bounds on the localization length of eigenfunctions.
Our results for $d=2$ are the first quantum diffusion results for the Anderson model on $\bbZ^2$. The proof avoids the use of diagrammatic expansions and instead proceeds by analyzing certain self-consistent equations for $R(z)$. This is facilitated by new estimates for $\|R(z)\|_{\ell^p\rightarrow \ell^q}$ that control the recollisions.
\end{abstract}
\maketitle
\section{Introduction}

In this paper, we consider the Anderson tight-binding model
\begin{equation}
\label{eq:tbm}
H=\Delta_{\Z^d}+\lambda V,
\end{equation}
where $\Delta_{\Z^{d}}$ is the nearest-neighbor Laplacian on $\Z^{d}$ (normalized as in~\eqref{eq:laplacian-def}), the potential $V:\Z^{d}\to \Real$ has independent standard Gaussian values, and $\lambda>0$ is a coupling constant.

The Hamiltonian~\eqref{eq:tbm} was originally introduced by Anderson~\cite{anderson1958absence} in order to study the propagation of electrons in a disordered environment, and more generally functions as a simple model for wave transport in random media. For $\lambda\gg 1$ or $d=1$, there is a large body of literature showing that $H$ exhibits \textit{Anderson localization}, meaning that almost surely it has an orthonormal basis of exponentially decaying eigenfunctions, see \cite{aizenman2015random} and the references therein.
In contrast, the small coupling regime $\lambda\ll 1$ is the subject of several outstanding conjectures in mathematical physics (\cite{simon1984fifteen,simon2000schrodinger}).
Referred to as the \emph{extended states} and \emph{quantum diffusion} conjectures, they predict the existence of delocalized states whose Schr\"{o}dinger evolution is described by a diffusion process.
We list a few forms of these conjectures informally,
which are expected to hold for $d\geq 3$ and sufficiently small $\lambda$:
\begin{itemize}
\item \emph{Diffusion of the propagator}: For fixed $\psi\in\ell^2(\bbZ^d)$ the probability density of the wavefunction $|(e^{-itH}\psi)(x)|^{2}$ satisfies a heat equation under an appropriate scaling limit.
In particular, if $\psi$ is localized, then
\begin{equation}
\label{eq:dynamical-deloc}
\sum_{x\in\Z^d} |x|^2 |(e^{-itH} \psi)(x)|^2 \sim \lambda^{-2}t ,
\end{equation}
implying that $e^{-itH}\psi$ has most of its $\ell^2$ mass at the \textit{diffusive} scale $\lambda^{-1}\sqrt{t}.$
\item \emph{Diffusion of the resolvent:}
For $E$ in the interior of the spectrum $\sigma(\Delta)$, the squared entries of the resolvent $R(z):=(H-E-i\eta)^{-1}$ satisfy an elliptic equation under an appropriate scaling limit as $\eta\to 0$.  In particular, they scale like the Green's function of the Laplacian:
\begin{equation}
\label{eq:resolvent-scaling-intro}
|R(z)_{xy}|^2 \sim\lambda^2 |x-y|^{2-d}.
\end{equation}
\item \emph{Delocalization of eigenfunctions:} Solutions of $H\psi = E\psi$ are \textit{delocalized} for energies $E$ in the interior of $\sigma(\Delta)$.
On a qualitative level, this means that the spectrum is continuous in the bulk.
Quantitatively, this should be interpreted to mean that most eigenfunctions of some truncation of $H$ satisfy a lower bound on their ``localization length,"
i.e., one of several measures of their spatial extent.
\end{itemize}

The factors of $\lambda^{-2}$ and $\lambda^2$ in the above are significant because they are, respectively, the critical time and energy scales past which diffusion may first be observed. Since the propagator and the resolvent are related by a Fourier transform in the time/energy variables, the first two items are thought of as equivalent, although in practice establishing the correspondence is non-trivial. In dimension $d=2$, it is instead conjectured that the spectrum of $H$ is pure point with dynamically localized eigenfunctions.
Nevertheless, in this dimension, diffusion is still expected up to the timescale $e^{\lambda^{-2}}$ (or dually down to the energy scale $\Impt z\sim e^{-\lambda^{-2}}$).

Progress on these conjectures, reviewed below, has been limited to a handful of results
and has mostly focused on establishing a scaling limit for the propagator \emph{in expectation} and at timescales just beyond the diffusive threshold of $\lambda^{-2}$.
In this paper, we show that for $d\geq 2$ the resolvent $R(z)$ is diffusive down to energy scales $\Impt z\sim\lambda^{2+\kappa_d}$, $\kappa_d>0$, \emph{with high probability}.
As consequences, we establish quantum diffusion (in a time-averaged sense) for the propagator at the longest timescale known to date and improve the best available lower bounds on the localization length of eigenfunctions.
Our results for $d=2$ are the first quantum diffusion results for the Anderson model on $\bbZ^2$.

We now state our main results before comparing them more closely to prior work.  Below we use $(e^{-itH})_{xy}$ and $R(z)_{xy}$ to
refer to the entries of $e^{-itH}$ and $R(z)$ considered as matrices. Note that $H$ is self-adjoint as the sum of a bounded operator and a multiplication operator, so there is no trouble defining its propagator.
\begin{theorem}[Dynamical delocalization]
\label{thm:main-prop}
Let $\delta>0$ and $H$ be as in \eqref{eq:tbm}.
There exists $c_{\delta},C_{\delta}>0$,  independent of $\lambda$,  such that for any $C_{\delta}\leq T\leq \lambda^{-2-\kappa_d+\delta}$ we have
\begin{equation}
\label{eq:averaged-prop-in-ball}
\frac{1}{T}\int_0^{T} \sum_{x\in \mathbb{Z}^d} |x|^2|(e^{-itH})_{0x}|^2
\geq c_{\delta}\lambda^{-2}T
\end{equation}
and
\begin{equation}
\label{eq:lower-bound-prop}
\frac{1}{T}\int_0^{T} \sum_{|x|\geq C_{\delta}\lambda^{-1}\sqrt{T}} |(e^{-itH})_{0x}|^2\leq \delta,
\end{equation}
with probability at least $1-C_{\delta}\lambda^{1000}$ and $\kappa_d>0$ given by
\begin{align*}
\kappa_d =
\begin{cases}
\frac{2}{13}, &d=2, \\
\frac{2}{9}, &d=3 \\
\frac12, &d >12.
\end{cases}
\end{align*}
\end{theorem}
See the statement of Theorem~\ref{thm:diffusive-profile} for the precise values of $\kappa_d$ (including values for $4\leq d\leq 12$).
Next we state our result about the diffusive scaling of the resolvent. For technical reasons, explained in Section~\ref{sec:prelims}, we exclude the following discrete set of energies (in addition to the energies outside of the spectrum of $\Delta$)
\[
\Sigma_d := (\Real \setminus [-2d,2d]) \cup (2d+4\bbZ)  \cup \{0\}.
\]
\begin{theorem}\label{thm:scaling}
Let $d\geq 3$, $E\not\in\Sigma_d$, and $\kappa < \kappa_d$.
Then there is a constant $\beta_E$ and $\kappa'<\kappa$ such that the following holds:
For any $f\in C_c^\infty(\Real^d)$ and setting $\ell = \lambda^{-2-\kappa'/2}$ we have, almost surely,
\begin{equation}
\label{eq:resolvent-limit}
\lim_{\lambda\to 0}
\lambda^{-2} \ell^{-2}
\sum_{x\in\bbZ^d} f(x/\ell) |R(E+i\lambda^{2+\kappa})_{0x}|^2
=  c_d \beta_E \int f(y) |y|^{2-d} \diff y.
\end{equation}
The analogous result holds in $d=2$ provided $f=\Div \vec{v}$ for a vector field
$\vec{v}\in C_c^\infty(\Real^d;\Real^d)$ and $|y|^{2-d}$ is replaced by $-\log(|y|)$.
\end{theorem}

The above result establishes an averaged version of the scaling~\eqref{eq:resolvent-scaling-intro}.  Indeed, rearranging the above equation we obtain
\[
\sum_{x\in\bbZ^d} f(x/\ell)|R_{0x}|^2
= \lambda^2 c_d\beta_E \int f(y/\ell) |y|^{2-d}\diff y
+ o(\lambda^2\ell^2).
\]
To understand the scaling of $\ell$, note that with $\eta := \Impt z = \lambda^{2+\kappa}$, we have
$\ell \ll \lambda^{-1}\eta^{-1/2}$.
The scale $\lambda^{-1}\eta^{-1/2}$ is the analogue of the diffusive length scale $\lambda^{-1}T^{1/2}$ for the resolvent under the correspondence $T \leftrightarrow \eta^{-1}$.

Finally, we state our result on the delocalization of eigenfunctions. We consider $H_L$, as in \eqref{eq:tbm},
but restricted to the torus $\bbZ^d_L := \bbZ^d/ (L\bbZ^d)$ with $L =\left\lceil\lambda^{-100}\right\rceil$ (fixed throughout the paper).  We take $\psi_j$ to be an orthonormal basis\footnote{With probability $1$, the eigenvalues are distinct.} of eigenfunctions with eigenvalues $E_j$ so that $H_L\psi_j = E_j\psi_j$.
For any $r>0$, we define the set of localized eigenvalues by
\begin{align*}
\mathfrak{L}(r) & :=
\{E_j \mid \text{ there exists } x_0\in\bbZ^d_L \text{ such that }
\|\psi_j\|_{\ell^2(B_r(x_0))} \geq 1 -r^{-4d}\}.
\end{align*}
Note $\mathfrak{L}(r)$ contains all eigenvalues whose eigenfunction  exponentially decays at scale $\ll r$.
\begin{theorem}[Delocalization bound for most eigenfunctions]
\label{thm:deloc}
Fix any $\eps,\delta>0$ and $E\in [-2d,2d]$ with $d(E,\Sigma_d)\geq 2\eps$.
Then with probability at least $1-\lambda^{1000}$,
\begin{align*}
\frac{\Big | \mathfrak{L}(\lambda^{-2-\kappa_d/2+\delta})\cap [E-\eps, E+\eps]\Big |}{\Big |\sigma(H)\cap [E-\eps, E+\eps]\Big|} \lsim_{\epsilon,\delta} \lambda^{c\delta}
\end{align*}
with $\kappa_d$ as in Theorem~\ref{thm:main-prop}.
\end{theorem}
In fact, one can take $\eps = \lambda^c$ above for $c$ small enough depending on $\delta$.  Moreover, we deduce Theorem~\ref{thm:deloc}
from a bound on the number of eigenfunctions localized near any specific $x_0\in\bbZ^d_L$, see Section~\ref{sec:deloc-pf}.

All prior quantum diffusion results on the Anderson model are formulated for the propagator and proceed by expanding $e^{-itH}$ perturbatively as a Dyson series.
This line of attack began with the
work of Spohn~\cite{spohn1977derivation} and was continued by Erd\H{o}s and Yau~\cite{EYkinetic}, who proved a scaling limit for the expected Wigner  transform\footnote{This encodes the distribution in phase space of the wavefunction $\psi$, see~\cite{EYkinetic} for the precise definition.} of $\psi$
up to the kinetic time $\lambda^{-2}$, which thus implies a version of~\eqref{eq:dynamical-deloc} in expectation at that time.  The combinatorial explosion of diagrammatic terms involved in the computation and the need to account for cancellations between the diagrams makes progress along these lines extraordinarily difficult.  Nevertheless, in a technical tour de force, Erd{\H{o}}s, Salmhofer, and Yau~\cite{ESYAnderson, ESYrecollision} established a scaling limit for the expected Wigner transform which again implies~\eqref{eq:dynamical-deloc}  up to a diffusive timescale on $\bbZ^d$, $d\geq 3$ (see also~\cite{ESYRSE,hernandez2024quantum} for analogous results on $\Real^d$).
In comparison to these results, our Theorem~\ref{thm:main-prop} holds with high probability, establishes diffusion up to a much longer timescale (on $\bbZ^3$,  up to $T=\lambda^{-2-\frac29}$ compared to $T=\lambda^{-2-\frac{1}{9800}}$ in~\cite{ESYAnderson}), and treats the hitherto unknown setting of $\bbZ^2$.  The drawback is that we do not prove a scaling limit and we require an average in time, which is a consequence of the fact that we compute primarily with the resolvent instead of the propagator.

For the resolvent, our main object of study, Theorem \ref{thm:scaling} is, to our knowledge, the first available in the diffusive energy regime,
though it is likely that the methods of~\cite{ESYAnderson} could be extended to obtain a diffusive result for $\Expec |R_{xy}(z)|^2$ when $\Impt z>\lambda^{2+\frac{1}{9800}}$.

Mathematically rigorous analysis of delocalization of eigenfunctions began with
the work of Schlag, Shubin, and Wolff~\cite{SSW}.  They showed that for the tight-binding
model on $\bbZ^2$, eigenfunctions must be delocalized to balls of radius at least the kinetic scale $\lambda^{-2}$.  A similar result was obtained by T. Chen~\cite{Chen} on $\bbZ^d$ using diagrammatic estimates along the lines of~\cite{EYkinetic}.  In principle, Chen's argument can be combined with the bounds of~\cite{ESYAnderson}
to prove a localization lower bound on most eigenfunctions of $H$ on the order at least $\lambda^{-2-\frac{1}{19600}}$, but
this does not appear explicitly in the literature to our knowledge.  Aside from this hypothetical bound,
we are not aware of any delocalization lower bounds for eigenfunctions of $H$ beyond the kinetic scale.

While progress on delocalization for the Anderson model itself has been very limited, there have been several breakthroughs in understanding delocalization in a variety of random band matrix ensembles
(for a partial list, see~\cite{bourgade2019random, bourgade2020random, yang2021random, yang2021delocalization,yang2022delocalization,yang2025delocalization,dubova2025delocalization}).
In these random matrix models the asymptotic parameter that replaces the coupling strength $\lambda$ is the ``bandwidth'' $W$ that parametrizes the number of random elements per row.
Thus, they may be seen as interpolations between the Anderson Hamiltonian, which has randomness only on the diagonal, and better understood mean-field models, in which every entry is random.
At a very high level, the framework developed to analyze these ensembles proceeds by studying certain \emph{self-consistent equations} satisfied by averages of the resolvent.

To our knowledge, no works have appeared attempting to address the Anderson Hamiltonian~\eqref{eq:tbm} at weak coupling via these random matrix methods.  The difficulty in doing so stems not only from the comparatively small amount of randomness present in the Anderson model, but also from its structure. Indeed, the Anderson model enjoys fewer symmetries than random band models, and the dispersive nature of the background Laplacian affects the small $\lambda$ phenomena in an essential way.

In this paper, we show that the self-consistent equation framework can nevertheless be used to derive diffusion of the Anderson Hamiltonian, provided that one has suitable dispersive estimates for the resolvent $R(z)$.
The analogous estimates for the free resolvent $(\Delta_{\Z^d}-z)^{-1}$ follow from oscillatory integral estimates in~\cite{ErdosSalmhofer,Taira,CueninSchippa}.
A key novelty in this work is a perturbative argument based on~\cite{BDH} that allows us to extend these bounds to the \emph{full} resolvent $R(z)$ at the energy scale $\Impt z\sim \lambda^2$.
In fact, the dispersive estimate we prove may be seen as a strengthened \emph{crossing estimate}, which is one of the key technical tools used to derive diffusion, for example, in~\cite{ESYAnderson}.
Remarkably, this allows us to completely sidestep the combinatorial difficulties inherent to the diagrammatic approach pursued in \cite{ESYAnderson,hernandez2024quantum}, as we discuss in Section~\ref{sec:crossingdiscussion}.

The remainder of the introduction is organized as follows.  In Section~\ref{sec:sketch} we provide a sketch of the argument including some preliminary calculations which is aimed to present most of the main ideas in a pedagogical way.  In Section~\ref{sec:nodiagrams} we explain why our relatively simple argument can succeed without any diagrammatic expansions.  In Section~\ref{sec:classical} we
compare our work to work on classical models of diffusion.

\subsection{Main ideas and discussion}
\label{sec:sketch}
Each of Theorem~\ref{thm:main-prop}-\ref{thm:deloc} follow from appropriate estimates on the quantity
\[
(R(z)^* A R(z))_{00} = \sum_{x\in\bbZ^d} a(x) |R(z)_{0x}|^2
\]
for some $a\in\ell^\infty(\bbZ^d)$, where $A$ is the diagonal matrix with diagonal specified by $a$.
At a simplified level, there are three steps to understanding the quantity above.

The first step is to establish a concentration inequality for the entries of the resolvent via the aforementioned dispersive estimate, encoded as a bound on $\|R(z)\|_{\ell^1\to\ell^4}$.
In the second step, we derive a self-consistent equation for $\Expec R$ whose error is controlled by the concentration inequality.
The third step is to show that $(RAR^*)_{xx}$ satisfies an elliptic equation and from this derive diffusion.
\subsubsection{Step 1: Concentration inequalities for the resolvent.}
\label{sec:conc-ineq}
We think of $R(z)$ as a function of the Gaussian entries $g_x$ of $V$ and apply the Gaussian Poincar{\'e} inequality to bound the variance.
To facilitate the calculation we write
\[
V = \sum_{x\in\bbZ^d} g_x \ket{x}\bra{x},
\]
where $\ket{x}\bra{x}$ is the rank-one projection onto the site at $x$.  Then we have
\begin{equation}
\label{eq:resolvent-derivative}
\frac{\partial}{\partial g_x}R=-\lambda R\ket{x}\bra{x}R,
\end{equation}
so the Gaussian Poincar{\'e} inequality implies
\begin{align*}
\operatorname*{Var} R_{xy}
&\leq \sum_{w\in\bbZ^d} \Expec |\frac{\partial}{\partial g_w} R_{xy}|^2 \\
&= \lambda^2 \sum_{w\in\bbZ^d} \Expec |R_{xw}R_{wy}|^2 \\
&\leq \lambda^2
\Big(\Expec \sum_{w\in\bbZ^d} |R_{xw}|^4\Big)^{1/2}
\Big(\Expec \sum_{w\in\bbZ^d} |R_{wy}|^4\Big)^{1/2} \\
&\leq \lambda^2 \Expec \|R\|_{1\to 4}^4.
\end{align*}

The main new technical contribution of this work is a bound for $\|R(z)\|_{1\to 4}$ that is effective for $z=E+i\eta$ with $\eta \ll \lambda^2$ in the
diffusive regime. In particular, we will show in  Proposition~\ref{prp:apriori}
\begin{equation}
\label{eq:simplified-onefour}
\|R(z)\|_{1\to 4} \leq  \lambda^{2-c_d} \eta^{-1}.
\end{equation}
with high probability so that for some range of $\eta$ much smaller than the diffusive scale $\lambda^2$, we still have $\operatorname*{Var} R_{xy} \ll 1$.  A similar argument can be used to bound higher moments of $R_{xy}-\Expec R_{xy}$.

\subsubsection{Step 2: The self-consistent equation for $\Expec R$}
\label{sec:sce-intro}
In this section we compute a self-consistent equation for $\Expec R$ by using the ``resolvent method" explained in~\cite{erdos2019matrix}.

The idea is to first renormalize by writing
\[
\Delta + \lambda V - z = (\Delta - (z+\lambda^2\theta)) + (\lambda V + \lambda^2\theta),
\]
for $\theta$ a scalar to be determined.
Setting $M_\theta := (\Delta - (z+\lambda^2\theta))^{-1}$, we have by the resolvent identity
\begin{equation}
\label{eq:resolvent-id}
R = M_\theta - M_\theta (\lambda V + \lambda^2 \theta) R.
\end{equation}
We now take an expectation on both sides of~\eqref{eq:resolvent-id}, and use
Gaussian integration by parts to compute the term involving $\Expec VR$.
Indeed using $V = \sum_x g_x \ket{x}\bra{x}$ and~\eqref{eq:resolvent-derivative} we compute
\begin{align*}
\Expec VR &=\sum_{x\in \Z^{d}_L}\ket{x}\bra{x}\Expec[g_xR]\\
& = -\lambda \sum_{x\in \Z^{d}_L}\ket{x}\bra{x}\mathbb{E}[R\ket{x}\bra{x}R].
\end{align*}

Introducing the superoperator
\[
\mathcal{D}:\mathcal{B}(\ell^2(\mathbb{Z}^d_L))\to \mathcal{B}(\ell^2(\mathbb{Z}^d_L))
\]
defined by
\begin{equation}
\label{eq:D-def}
\mathcal{D}[A] := \sum_{x}A_{xx}\ket{x}\bra{x},
\end{equation}
we have thus shown
\begin{equation}
\label{eq:sce-first-draft}
\mathbb{E}R = M_\theta(\Id+\lambda^2\mathbb{E}\mathcal{D}[R-\theta \Id]R )
\end{equation}
Note we can choose any value of $\theta$, so now we set $\tilde{\theta}(z) := \Expec R_{00}(z)$, so that $\theta\Id=\Expec \mcal{D}[R]$ by translation invariance.  Then writing
$\tilde{M} = \tilde{M}(z) = M_{\tilde{\theta}(z)}(z)$, we have arrived at the following system of equations.
\begin{proposition}
\label{prp:approxSCE}
With $\tilde{\theta}(z) := \Expec R_{00}(z)$ and $\tilde{M}$ as defined above, we have
\begin{equation}\tag*{(SCE)}
\label{eq:appx-sce}
\begin{split}
&\mathbb{E}R = \tilde{M}(\Id+\mathfrak{E}_{\rm loc}) \\
&\mathfrak{E}_{\rm loc} := \lambda^2\mathbb{E}\mathcal{D}[R-\Expec R]R.
\end{split}
\end{equation}
\end{proposition}
If $\mathfrak{E}_{\rm loc} \approx 0$, we could approximate $\tilde{\theta} = \Expec R_{00}$ by taking the diagonal entry of either
side of~\ref{eq:appx-sce} to obtain that $\tilde{\theta}$ solves an approximate version of the self-consistent equation
\begin{equation}
\label{eq:theta-appx-sce}
\theta = (\Delta - (z+\lambda^2 \theta))^{-1}_{00}.
\end{equation}

In Theorem~\ref{thm:local-law} we perform an error analysis to show that indeed $\tilde{\theta}$ is close to the solution $\theta$ of the equation above.
To do this we first note  that $\mcal{D}[R-\Expec R]$ is a diagonal
matrix with entries $R_{xx}-\Expec R_{xx}$.   These entries are controlled by Step 1, and thus $\mathfrak{E}_{\rm loc}$ is small.  To conclude, we analyze the stability of the self-consistent equation~\eqref{eq:theta-appx-sce} to ensure that the small perturbation $\mathfrak{E}_{\rm loc}$ also changes $\tilde{\theta}$ by a small amount.

\subsubsection{Step 3: The diffusive profile of $R$}
\label{sec:Teqintro}
The approximate self-consistent equation~\ref{eq:appx-sce} can be used to estimate $\Expec R_{xy}$, but this is insufficient to compute averages of the form $\Expec\sum_{y\in\Z^d}a(y)|R_{xy}|^2$.
For this, we follow
the \textit{T-equation} strategy developed by Erd\H{o}s, Knowles, Yau and Yin~\cite{Y3Teq} in the context of random band matrices.

Again using the resolvent identity~\eqref{eq:resolvent-id} and Gaussian integration by parts, one may derive the T-equation
\begin{equation}
\label{eq:T-eq}
\Expec RAR^* = \tilde{M}A\tilde{M}^* + \lambda^2\tilde{M}\mathcal{D}(\Expec RAR^*)\tilde{M}^* + \mathfrak{E}_{\rm T},
\end{equation}
where $\mathfrak{E}_{\rm T}$ is an error term specified in Lemma~\ref{lem:T-eq-full}.

Writing $f(x) = \Expec (RAR^*)_{xx}$, the T-equation implies
\begin{equation}
\label{eq:elliptic-pde-intro}
f = \tilde{K} \ast a + \lambda^2 \tilde{K} \ast f + \mathfrak{e}
\end{equation}
where $\mathfrak{e}(y) = (\mathfrak{E}_{\rm T})_{yy}$ and the convolution kernel $\tilde{K}$ is given by
\begin{equation*}
\tilde{K}(x) = |\Braket{0| \tilde{M}|x}|^2.
\end{equation*}
Let $\tilde{\mcal{K}} u = \tilde{K}\ast u$ be the convolution operator with kernel $\tilde{K}$.  Then
\begin{equation*}
\begin{aligned}
(\Id - \lambda^2\tilde{\mcal{K}}) f
&=  \mcal{\tilde{K}}a + \mathfrak{e}.
\end{aligned}
\end{equation*}
As we explain in Section~\ref{sec:RAR}, the scaling limit of the equation above is an elliptic PDE.
This is because $\sum \lambda^2 \tilde{K}(x) \approx 1$ so that $(\Id - \lambda^2\tilde{\mcal{K}})g$ is the difference
between $g$ and a local symmetric average of $g$. Therefore, it approximately acts as an elliptic operator when applied
to functions that are smooth to scale $\lambda^{-2}$.  Alternatively, by iterating~\eqref{eq:elliptic-pde-intro} one obtains an expansion for $f$ that can be interpreted in terms of a random walk.

\subsection{Diagrammatics and the $\|R\|_{1\to 4}$ estimate.}
\label{sec:nodiagrams}
The proof sketched above is most notable for a feature it \textit{lacks} -- any mention of diagrammatic
expansion.  There are two reasons we are able to sidestep diagrammatic techniques with our argument.
The most straightforward reason is that we adapt the framework of~\cite{Y3Teq} to perform our calculations, using``self-consistent equations'' for the
moments of the resolvent to circumvent perturbative series expansions.
What allows our calculations to go through at a technical level however is our estimate on $\|R\|_{1\to 4}$, and we explain in this section how this encodes a bound on all of the recollision diagrams
of~\cite{ESYrecollision}.

\subsubsection{Formal perturbation theory and renormalization}
\label{sec:formal-perturbation}
First we briefly review the traditional diagrammatic approach to understanding the Anderson model, following a formal perturbative calculation of Spencer~\cite[Chapter 1, Appendix B]{frohlich2012quantum}.
For simplicity we work with the expected resolvent $\Expec R(z)$.
First we write $R(z)$ as a Born series expansion, perturbing around the free resolvent $R_0(z) := (\Delta - z)^{-1}$:
\[
\mathbb{E}R(z) = R_0(z) + \sum_{k=1}^\infty (-\lambda)^k\mathbb{E} R_0(z) (VR_0(z))^{k}.
\]
When $\eta:= \Impt(z)\gg\lambda^2$, the series converges and is dominated by the first term which is of order $\lambda \eta^{-1/2}$.
To prove this convergence rigorously, one can establish that the bound
\begin{equation}
\label{eq:sqrt-R}
\|R_0^{1/2}(z)VR_0^{1/2}(z)\|_{2\to 2}\lsim \eta^{-1/2}
\end{equation}
holds with high probability.
A bound along these lines of~\eqref{eq:sqrt-R} follows from the arguments in any of~\cite{bourgain2001random, SSW, BDH} (though see also the computations in~\cite{van1954quantum, davies1974markovian, spohn1977derivation, EYkinetic} for alternative explanations of this square-root cancellation phenomenon).

If $\eta \lsim \lambda^2$ the sum over $k$ is no longer perturbative and one needs to incorporate cancellations in the sum. Omitting $z$ from notation and writing $\ket{x} = \delta_x$ and $\ket{x}\bra{x}$ for the corresponding rank-one projection we can write the next non-vanishing term as follows
\[
\lambda^2 \mathbb{E}R_0VR_0VR_0
= \lambda^2 \sum_{x\in \mathbb{Z}^d}R_0\ket{x}\braket{x|R_0|x}\bra{x}R_0= \lambda^2(R_0)_{00}R_0^2.
\]
The right hand size has operator norm $\eta^{-1}(\lambda^2\eta^{-1})$
which is much larger than $\eta^{-1} = ||R_0||_{2\to 2}$ if $\eta\ll\lambda^2$.
Similar calculations for larger $k$ show that to sum the series one needs to take advantage of  cancellations.
This is done by renormalizing, that is, by perturbing around $M_\theta = (\Delta- (z+\lambda^2\theta))^{-1}$ rather than $R_0(z)$.
Performing a Born series expansion about $M_\theta$ we obtain
\begin{equation}
\label{eq:renormalized-born}
\mathbb{E}R = M_\theta + \sum_{k=1}^\infty \Expec M_\theta((\lambda V + \lambda^2 \theta)M_\theta)^k
\end{equation}
Now we choose $\theta$ so that the second order terms in $\lambda^2$ vanish:
\[
\lambda^2\tilde{R}_0\theta \tilde{R}_0
- \lambda^2 \mathbb{E}\tilde{R}_0 V\tilde{R}_0V\tilde{R}_0 = 0.
\]
By the computation for $k=2$ above this is equivalent to taking $\theta_{\rm sce}(z)$ to solve the self-consistent equation
\begin{equation}
\label{eq:theta-sce}
\theta_{\rm sce}(z) = (\Delta-z-\lambda^2\theta_{\rm sce}(z))^{-1}_{00},
\end{equation}
which we observe is the same equation derived above~\eqref{eq:theta-appx-sce}.

With this renormalization, one expects that (setting $M_{\rm sce} = M_{\theta_{\rm sce}}$),
\begin{equation}
\label{eq:ER-guess}
\mathbb{E}R \approx M_{\rm sce}.
\end{equation}
To establish~\eqref{eq:ER-guess} one could in principle attempt to bound the remaining terms in~\eqref{eq:renormalized-born}
and indeed, this is essentially the strategy of~\cite{ESYRSE,ESYAnderson}.

\subsubsection{Resolvent estimates as crossing integrals}
\label{sec:crossingdiscussion}
It is instructive to consider the first error term in~\eqref{eq:renormalized-born} that does not cancel exactly as a consequence
of the choice of $\theta_{\rm sce}$.  This term comes from the term in the Wick expansion of the fourth order term which involves a crossing, that is,
\begin{equation}
\label{eq:crossing-term}
\mcal{E}_{\rm cross} := \lambda^4 \sum_{x,y} M_{\rm sce} \ket{x}\bra{x} M_{\rm sce} \ket{y}\bra{y}M_{\rm sce}\ket{x}\bra{x} M_{\rm sce}
\ket{y}\bra{y} M_{\rm sce}.
\end{equation}
The operator $M_{\rm sce}$ is a Fourier multiplier with symbol
\[
m(p) = \frac{1}{\omega(p) - (z+\lambda^2\theta_{\rm sce})},
\]
so writing out the entries of $\mcal{E}_{\rm cross}$ using the Fourier transform we obtain
\begin{align*}
\braket{w_1 | \mcal{E}_{\rm cross} | w_2}
&= \lambda^4  (2\pi)^{-4d} \iint_{(\Torus^d)^5} \sum_{x,y} e^{i(w_1-x)\cdot q_1} e^{i(x-y)\cdot (q_2-q_3+q_4)} e^{i(y-w_2)\cdot q_5}
\prod_{j=1}^5 m(q_j) \diff \vec{q} \\
&= \lambda^4  (2\pi)^{-4d} \iint_{(\Torus^d)^4} \delta(-q_1+q_2-q_3+q_4)
e^{i(w_1-w_2) \cdot q_1} m(q_1) \prod_{j=1}^4 m(q_j) \diff \vec{q}.
\end{align*}
Thus one has
\[
\max_{w_1,w_2} |\braket{w_1 | \mcal{E}_{\rm cross} | w_2}| \lsim
\lambda^4 (\Impt (z + \lambda^2\theta_{\rm sce}))^{-1}
I_4(z+\lambda^2\theta_{\rm sce}),
\]
where $I_4(z)$ is the
\textit{four denominator crossing integral} (see~\cite{ErdosSalmhofer,lukkarinen2007asymptotics})
\[
I_4(z) :=
\int_{(\Torus^d)^4}
\delta(q_1-q_2+q_3-q_4)\prod_{j=1}^4 \frac{1}{|\omega(q_j)-z|}
\diff \vec{q}.
\]
We note that similar ``crossing integrals''
appear in the derivation of kinetic limits of random wave equations~\cite{lukkarinen2007kinetic}
and also appear in the wave kinetic theory of nonlinear problems~\cite{deng2021derivation,staffilani2021wave}.  Note that
there is a trivial bound $I_4(z) \ll \eta^{-1}$ which comes from using the $L^\infty$ norm on $\frac{1}{|\omega(q_4)-z|}$ and using Fubini's
theorem for the integrals on $q_1$, $q_2$, and $q_3$.  This trivial estimate shows that the maximum entry of $\mcal{E}_{\rm cross}$
is $O(1)$.  Thus any nontrivial estimate on $I_4$ shows that the first error term is of subleading order.   In~\cite{ESYAnderson}
an intricate diagrammatic argument is developed to show that indeed a nontrivial estimate on $I_4(z)$ implies that the
sum over the remaining error terms is also of subleading order.

On the other hand one has by Plancherel's theorem the identity
\[
\sum_{x\in\bbZ^d} |\Braket{x|(\Delta_{\Z^d}-z)^{-1}|0}|^4
= \int_{(\Torus^d)^4}
\frac{\delta(p-q+r-v)\diff p\diff q \diff r \diff v}
{(e(p)-z)(e(q)-\bar{z})(e(r)-z)(e(v)-\bar{z})}.
\]
That is, the four-denominator estimate controls the $\ell^1\to \ell^4$ norm of the resolvent of the Laplacian,
\[
\| (\Delta_{\Z^d} -z)^{-1}\|_{\ell^1\to\ell^4}^4\leq I_4(z).
\]
From this perspective, our estimate on $\|(H-z)^{-1}\|_{\ell^1\to \ell^4}$  is a strict strengthening of the crossing estimate that~\cite{ESYAnderson} uses.
One way to interpret our approach is that we choose not to perform a perturbative expansion all the way down to the free resolvent but instead
use bounds on the full resolvent $(H-z)^{-1}$ itself.  This is what allows us to stop the expansion and directly estimate the error term.

\subsubsection{As a dispersive estimate}
Another way of understanding the quantity $\|R\|_{1\to 4}$ is as a dispersive estimate.  A calculation using Plancherel's theorem and
the Cauchy-Schwartz inequality shows that
\begin{equation}
\label{eq:onefourrecollision}
\begin{split}
\sum_{x\in\bbZ^d} \lambda^2 \Big(\int_0^{\eta^{-1}} |e^{itH} \psi(x)|^2 \diff t \Big)^2
&\leq e^2
\sum_{x\in\bbZ^d} \lambda^2 \Big(\int_0^\infty e^{-\eta t} |e^{itH} \psi(x)|^2 \diff t \Big)^2  \\
&= e^2
\sum_{x\in\bbZ^d} \lambda^2 \Big(\int |R(E+i\eta)\psi(x)|^2 \diff E \Big)^2  \\
&\lsim \lambda^2 \Big(\int \|R(E+i\eta)\psi\|_{\ell^4}^2 \diff E\Big)^2.
\end{split}
\end{equation}
To interpret the left hand side, note that $\int_0^T |e^{itH}\psi(x)|^2\diff t$ measures the total probability of finding the particle at
location $x$ for some time in $[0,T]$, and so $\lambda \int_0^T |e^{itH}\psi(x)|^2\diff t$ measures the contribution to the probability amplitude
 that the wavefunction $\psi$ ever interacts with the potential at site $x$ within time $[0,T]$.  The square then heuristically measures the
probability that the particle revisits a given site $x$ \textit{twice} within $[0,T]$.  Thus, the left hand side measures the contribution from
all recollisions up to time $T$.

Now taking $\psi = \delta_0$, one observes that $\|R(E+i\eta)\delta_0\|_{\ell^4} \leq \|R(E+i\eta)\|_{\ell^1\to\ell^4}$.  In conclusion, $\|R\|_{1\to 4}$ can be used to estimate the likelihood that an initially localized wavefunction experiences any recollision up to time $T=\eta^{-1}$.  That is, it encodes the sum over all recollision diagrams in the sense of~\cite{ESYrecollision}.

From another perspective, the inequality above is some kind of dispersive estimate, which in particular shows that $\psi_t$ does not interact
too strongly with any particular potential site.  Therefore, one expects a concentration phenomenon for $\psi_t(x)$.  This suggests some heuristic
that concentration inequalities and crossing estimates should be related, as they both come from estimating the strength of interaction with single
sites.

\subsection{Comparison to classical models}
\label{sec:classical}

The Anderson model is a discrete cousin of the random Schrodinger equation (not addressed in this paper),
\[
i\partial_t \psi = -\frac12 \Delta_{\Real^d} \psi + \lambda V\psi.
\]
The semiclassical limit of this equation, corresponding to initial data $\psi_0$ which is highly oscillatory,  is the stochastic acceleration model
\begin{equation}
\label{eq:stochastic-acceleration}
\frac{d^2}{dt^2} x = -\lambda \nabla V(x).
\end{equation}
A diffusive limit for stochastic acceleration was first established in $d>2$ by Kesten and Papanicolaou~\cite{kesten1980limit} (and see also~\cite{durr1987asymptotic,komorowski2006diffusion} for results in $d=2$).  The idea in this proof is to establish a coupling between the motion described
by~\eqref{eq:stochastic-acceleration} and a diffusive Markov process.  This coupling can be established easily up until the first time that the particle trajectory intersects itself.  Similar analysis can be performed on other models of walks in random environments (see for example~\cite{boldrighini1983boltzmann} for an analysis on the wind-tree model).
The diffusive limit of stochastic acceleration has also been used directly to obtain a quantum diffusion result for highly oscillatory initial data~\cite{bal2003self}.

In the Anderson model, the self interactions that break the Markov property correspond to crossing diagrams
in the expansion~\eqref{eq:renormalized-born} (which is also the error term in the self-consistent equation).  Our proof proceeds by simply bounding the total contribution of these terms (in the form of a bound on $\|R\|_{1\to 4}$, see~\eqref{eq:onefourrecollision} above), so in this sense is analogous to the ``Kesten-Papanicolaou argument''. In particular, it is limited to the timescale $\lambda^{-4}$ at which
point the crossing diagrams become nonnegligible without further renormalization.  We note that previous approaches to quantum diffusion~\cite{ESYRSE, ESYAnderson,hernandez2024quantum} required bounds on \textit{all} crossing diagrams to reach the diffusive timescale.  In the analogy with the classical problems, this would correspond to bounding not just the contribution of a single self-interaction but also the contributions from all possible $K$-fold self-interactions for $1\leq K\leq \lambda^2T$.   Thus the innovation of this paper is that we manage to complete the analysis without considering these higher order self-interactions.

More recent work on classical models of diffusion improve upon the Kesten-Papanicolaou framework by generalizing the notion of the coupling to a Markov process and by ``renormalizing'' the coupled Markov process.  In~\cite{lutsko2020invariance} a longer diffusive timescale is reached by allowing self-interactions but controlling their effect on the endpoint of the walk.  To reach even longer timescales some renormalization is necessary, in the sense that one modifies the limiting Markov process as one progressively reaches longer timescales.  An analysis to this effect was used in~\cite{elboim2022infinite} on the interchange process to reach infinite timescales, resolving a long-standing conjecture of Toth.    More recently, renormalization has been applied to a discrete model more analagous
to stochastic acceleration~\cite{elboim2025lorentz}.

We stress that at this point it is not clear how to execute a renormalization argument analogous to that of~\cite{elboim2022infinite,elboim2025lorentz} (in the classical setting) or of~\cite{yang2021delocalization,yang2021random,yang2021delocalization,dubova2025delocalization} (in the random matrix setting) to reach superpolynomial time scales in the Anderson model.  However, we are
hopeful that ideas in the present work might prove useful to do so.

\subsection*{Organization of the paper}
The rest of the paper is organized as follows.  In Section~\ref{sec:prelims} we lay out notation and some simple facts that will be used throughout. In Section~\ref{sec:resolvents} we prove bounds on $\|R\|_{p\to q}$.  Then in Section~\ref{sec:sce} we prove bounds on the entries of the resolvent using the self consistent equation for $R$.  This is applied in Section~\ref{sec:diffusion-profile} to perform an error analysis in the T-equation.  Then in Section~\ref{sec:RAR} we derive a comparison between the T-equation and an elliptic PDE on $\Real^d$.  Finally, in Section~\ref{sec:proofs} we use these results to complete the proofs of the main theorems.

\subsection*{Acknowledgements}
The authors would like to thank Antoine Gloria, Toan Nguyen, Lenya Ryzhik, Wilhelm Schlag, Charles Smart, Tom Spencer, and Gigliola Staffilani for clarifying discussions and encouragement.  This material is based upon work supported by the National Science Foundation under Awards No. 2503339 and 2303094.

\section{Preliminaries and Conventions}
\label{sec:prelims}
In this section we collect the notational conventions used throughout the paper, as well as some preliminary results that are proven in the appendix and used throughout the proof.

\subsection{The Laplacian and the Fourier transform}
We normalize the Laplacian as
\begin{equation}
\label{eq:laplacian-def}
(\Delta_{\bbZ^{d}}f)(x)=\sum_{|y-x|=1}f(y).
\end{equation}
For a function $f\in \ell^2(\bbZ^{d})$, we take its Fourier transform to be
\begin{align*}
	\hat{f}(\xi)=\sum_{x\in \bbZ^{d}}e^{ix\xi}f(x),
\end{align*}
and therefore
\begin{align*}
f(x)=(2\pi)^{-d}\int_{\bbT^{d}}e^{-ix\xi}\hat{f}(\xi)\:\diff\xi,
\end{align*}
where $\Torus^d = \Real^d / (2\pi \bbZ^d)$.
The symbol of $\Delta_{\Z^d}$ is thus
\begin{align*}
\omega(\xi):=\sum_{j=1}^{d}2\cos(\xi_j),
\end{align*}
so that its spectrum is $[-2d,2d]$.
The critical points of $\omega$ correspond
to values of $\xi$ for which $\xi_j \in\{0,\pi\}$ for each $j\in[d]$.  Therefore
the critical \textit{values} of $\omega$ is given by the set $\Sigma_{\rm crit} = [-2d,2d] \cap (2d+4\bbZ)$.  We define\footnote{Technically we only need to include $0$ in $d=3$, and the proof could be modified to take $\Sigma_d=\Real\setminus[-2d,2d]$ in $d\geq 4$, but excluding $\Sigma_{\rm crit}$ makes the analysis simpler.}
\[
\Sigma_d := (\Real \setminus [-2d,2d])\cup \Sigma_{\rm crit} \cup \{0\},
\]
with the value $0$ being included because in $d=3$ there is a flat umbilic point on the level set of $\{\omega = 0\}$ (as explained in~\cite{ErdosSalmhofer}).

\subsection{Density of states and the self-consistent equation}

The \textit{density of states} for the Laplacian
$\Delta_{\bbZ^d}$ is the measure $\diff \gamma$ such that
\[
\int_{\mathbb{R}} f(\sigma) \diff \gamma
= (2\pi)^{-d} \int_{\Torus^d} f(\omega(\xi))\diff\xi.
\]
As shown in Lemma~\ref{lem:free-dos}, $\diff \gamma = \rho(x)\diff x$ has a density supported on $[-2d,2d]$ which is strictly positive in $(-2d,2d)$.  Aside from a logarithmic singularity at $0$ in $d=2$,  the density is bounded.  The density $\rho$ also has the expression
\begin{equation}
\label{eq:rho-def}
\rho(E) := (2\pi)^{-d}\int_{\Torus^d}
\frac{1}{|\nabla \omega(\xi)|} \diff\xi.
\end{equation}

The significance of the function $\rho$ is that
\[
\rho(E) = \lim_{\eta\to 0} \frac{1}{\pi} \Impt ((\Delta_{\bbZ^d}-z)^{-1}_{00}).
\]
Since $(\Delta_{\bbZ^d}-z)^{-1}$ is analytic in $z$, its imaginary part is the harmonic extension of $\rho$ to the upper half plane.
 Similarly, its real part is the harmonic extension of $\rm{H}\rho$, where $\rm{H}$ is the Hilbert transform.

Finally, we note that the self-consistent
equation
\[
\theta = (\Delta_{\bbZ^d}-(z+\lambda^2\theta))^{-1}_{00}
\]
has a unique solution in the upper half-plane $\bbH$.  This is proven in Proposition~\ref{pr:thetaAPriori}.

\subsection{Truncation to $\Z^d_L$}
In Sections~\ref{sec:resolvents}-\ref{sec:diffusion-profile} we work with a truncation of the Hamiltonian $H$ to the torus $\Z^d_L:=\Z^d/(L\Z^d)$, with $L=2\left\lceil \lambda^{-100}\right\rceil$.
Specifically, let $\Delta_L$ be the nearest-neighbor Laplacian on this graph, normalized as in \ref{eq:laplacian-def}, and let $H_L=\Delta_L+\lambda V$, where by a slight abuse of notation $V$ denotes the potential restricted to the box $\Z^d\cap[-L/2,L/2]$ and thought of as periodic. Here and throughout, we identify elements of $\Z^d_L$ with $\Z^d\cap [-L/2,L/2]^d$ in the natural way.
Below, unlike in the introduction, $R(z)=(H_L-z)^{-1}$ denotes the resolvent of this truncated operator whereas $R_{\Z^d}(z)=(H-z)^{-1}$ is the full resolvent.

The large value of $L$ relative to the energy scale under consideration, $\Impt z>\lambda^{2+\kappa_d}$, allows us to pass freely between $R$ and $R_{\Z^d}$. By regarding elements of $\Z^d_L$ as elements of $\Z^d$ with $\ell^\infty$ norm at most $L$, we have the following:
\begin{proposition}
\label{prp:Zd-to-ZdL}
For any $x,y\in \Z^d_L$ with $|x|_\infty,|y|_\infty\leq L/4$ and $z=E+i\eta\in \bbH$ with $\eta>\lambda^{4}$, we have that
\begin{align*}
    \left|\braket{x|R(z)|y}-\braket{x|R_{\Z^d}(z)|y}\right|\leq C \exp(-c\lambda^{-10}),
\end{align*}
with constants that are uniform for all realizations of $V$.
\end{proposition}
\begin{proof}
Recall the Combes-Thomas estimate \cite[Thm 10.5]{aizenman2015random}, valid for any self-adjoint Schr\"{o}dinger operator on $\Z^d$: there exists dimensional constants $c,C>0$ such that if $|x-y|>C\eta^{-1}$ then
\begin{align}
\label{eq:CT}
|\braket{x|R_{\Z^d}(z)|y}|\leq C\exp(-c\eta|x-y|).
\end{align}
Observe that if either of $u,w\in \Z^d$ has $\infty$-norm less than $L/2$ then it is immediate from the definitions that
\begin{align*}
\braket{u|H_L|w}=\braket{u|H|w}.
\end{align*}
So, we write
\begin{align*}
\braket{x|R(z)|y}&=\sum_{w\in \Z^d_L}\braket{x|R(z)|w}\braket{w|(H-z)R_{\Z_d}|y}\\
&=\sum_{u\in \Z^d}\braket{u|R_{\Z_d}|y}\sum_{w\in \Z^d_L}\braket{x|R(z)|w}\braket{w|(H-z)|u}.
\end{align*}
If $|u|_\infty<L/2$, then we have that $\sum_{w\in \Z^d_L}\braket{x|R(z)|w}\braket{w|(H-z)|u}=\One_{x=u}$. Therefore, by using \eqref{eq:CT}
\begin{align*}
\braket{x|R(z)|y}&=\braket{x|R_{\Z^d}(z)|y}+\sum_{|u|_\infty=L/2}\braket{u|R_{\Z^d}(z)|y}\\
&=\braket{x|R_{\Z^d}(z)|y}+O(e^{-c\lambda^{-10}}),
\end{align*}
as desired.
\end{proof}

\subsection{Constants and asymptotic notation}
Throughout, $z=E+i\eta$ always denotes a complex parameter with real part $E$ and imaginary part $\eta$.

We regard $\eps$ and $\delta$ as fixed small parameters throughout, and we will always take $d(E,\Sigma_d)>\eps$.
We let $c$ and $C$ denote a constant that may change from instance to instance and depend only on $\delta$, $\eps$, and the dimension $d$. In particular, it is understood that they are independent of $\lambda$, $\eta$ and $z$.  We write $A\lsim B$ to mean that $A\leq CB$ for some large $C$ depending on $\eps$ and $\delta$, and $A\ll B$ to mean that $A \leq cB$ for a sufficiently small $c$ depending on $\eps$ and $\delta$.  For example, throughout the paper we think of $\lambda \ll 1$, meaning that $\lambda \leq \lambda_0(\eps,\delta)$. Subscripts such as $C_\delta$ or $\ll_\eps$ are merely meant to emphasize dependence.

It is also convenient to use a variant of \emph{stochastic domination} notation commonly used in random matrix theory, see, e.g. \cite{erdHos2017dynamical}.  For random quantities $A$ and $B$,
perhaps depending additionally on $z$ and $\lambda$, we write $A\ldom B$ (read as ``$A$ is stochastically dominated by $B$'') if for any small $\gamma>0$ and large $N>1$ there exists a constant $C(\gamma,N)$ (which, again, may also implicitly depend on $\eps$ and $\delta$) such that the event
\[
A \leq C(\gamma,N) \lambda^{-\gamma} B
\]
holds with probability at least $1-C(\gamma,N)\lambda^N$. When $A$ and $B$ depend on $z$ within a specified range, it is implicit that these constants are uniform in such energies.

\section{A priori estimates for the resolvent}
\label{sec:resolvents}
In this section we establish $\ell^p\to\ell^q$ mapping properties of the resolvent on $\Z^d_L$, denoted $R(z)=(H_L-z)^{-1}$.
To state the result we define the sequence of exponents $p_d$,
\begin{align*}
p_d=\begin{cases}
6&d\in\{2,4\}\\
\frac{14}{3}&d=3\\
\frac{2d}{d-3}&d> 4,
\end{cases}
\end{align*}
and we write $p_d' = \frac{p_d}{p_d-1}$.  We can now state the main result of the section.
\begin{proposition}[A priori resolvent bounds]
\label{prp:apriori}
For $p\leq p_d'<p_d\leq q$ and $z=E+i\eta$ with $\eta>\lambda^{10}$, $E\in[-2d,2d]$,
and $d(E,\Sigma_d)>\eps$, we have that
\begin{equation}
\label{eq:pq-apriori}
\|R(z)\|_{p\to q} \ldom \lambda^2\eta^{-1} + 1.
\end{equation}
Moreover, we have the bounds
\begin{align}
\label{eq:one-two-apriori}
&\|R(z)\|_{2\to q} = \|R(z)\|_{p\to 2} \ldom \eta^{-1/2}(\lambda^2\eta^{-1}+1)^{1/2},
\end{align}
with implicit constants depending only on $\eps$, $\delta$, and $d$.
\end{proposition}
We remark that the resolvent bounds above imply estimates on the $\ell^p$ norms of $\ell^2$-normalized eigenfunctions of $H$ on $\bbZ^d_L$.  In particular, if $H\psi = E\psi$ for $d(E,\Sigma_d)>\eps$, then $R(z)\psi = \eta^{-1}\psi$.  Therefore, for an $\ell^2$-normalized eigenfunction $\psi$ and  $q > p_d$ we can conclude that $\|\psi\|_{\ell^q} \ldom \lambda$.

There are two main ingredients involved in the proof of Proposition~\ref{prp:apriori}.  The first is an argument that allows us to transfer
$\ell^2\to \ell^q$ bounds (and, by duality, $\ell^p\to\ell^2$ bounds) from spectral projections of the Laplacian $\Delta_L$ to spectral
projections of the random Hamiltonian $H_L$.  This is based on the main result of~\cite{BDH}, which establishes a comparison in $\ell^2$
operator norm for these spectral projections that is effective down to intervals of width $\lambda^2$ (see Proposition~\ref{prp:zero-bd} below).  We can then recover bounds for the resolvent square-root $R^{1/2}$
by decomposing into spectral projections, and we conclude by estimating (for $p\leq 2\leq q$):
\[
\|R\|_{p\to q} \leq \|R^{1/2}\|_{p\to 2} \|R^{1/2}\|_{2\to q}.
\]

The next ingredient we need is an $\ell^2\to\ell^q$ bounds for the spectral projections of $\Delta$.  This is equivalent to a Tomas-Stein type restriction estimate~\cite{tomas1975restriction}, and follow from appropriate decay bounds for the Fourier transform of the uniform measure on the level sets of the dispersion relation.  In $d=2$ this is a simple consequence of the curvature of the level sets (for levels away from $\Sigma_2$), but in $d\geq 3$ this is made more complicated by the fact that these level sets have curves of vanishing Gaussian curvature.
The correct decay bound was established by Erd{\H o}s and Salmhofer~\cite{ErdosSalmhofer} for the
dispersion relation on $\bbZ^3$ up to logarithmic losses. In~\cite{Taira}, Taira established the correct pointwise decay without logarithmic losses, and a simpler proof was recently given by Cuenin and Schippa in~\cite{CueninSchippa}.  In dimensions $d>3$ we
are not aware of any work analyzing pointwise bounds for the decay of the Fourier transform of the level sets of the dispersion relation of the Laplacian.  Instead we use a simple argument in which we write the spectral projection as an integral
of the propagator $e^{it\Delta}$ and use a Strichartz estimate (as explained for example in~\cite[Section 5.4]{germain}).  This gives a suboptimal bound since we do not take into account cancellations in the time integral, but it suffices for our result.

A few simplifying comments are in order before we commence with the proof.  First, note that the stochastic
domination notation hides powers of $\lambda^{-\eps}$, so it suffices to consider $p<p_d'$ and $q>p_d$ (that is, we do not need to treat the endpoint exponents).  Second, we note that the free resolvent on $\bbZ^d$ is functionally equivalent to the free resolvent on $\bbZ^d_L$ (in
particular, their $\ell^p\to\ell^q$ mapping properties are the same)
by Corollary~\ref{cor:resolvent-truncation}.

\subsection{The transfer argument.}
\label{sec:transfer}

First, we record the following result comparing the spectral projections of $H$ to those of $\Delta$.
\begin{proposition}[Corollary 1.3 of~\cite{BDH}]
\label{prp:zero-bd}
Let $\chi \in C_c^\infty(\R)$ be a fixed smooth bump function.  Then for
$K\geq \sqrt{|\log \lambda|}$ the bound
\[
\|\chi( (H_L-E) / \alpha) - \chi( (\Delta_L-E)/\alpha) \|_{2\to 2}
\leq C K \lambda |\log \lambda|^2 \alpha^{-1/2}
\]
holds for all $E\in \R$ and $\delta>0$ with probability at least $1-e^{-cK^2}$.
\end{proposition}

This result follows directly from the proof of Corollary $1.3$ in~\cite{BDH}. That result is stated for a Hamiltonian on $\Z^d$ as in \eqref{eq:tbm}, but with the potential truncated to a box of size $L$. The proof can easily be adapted to the Hamiltonian $H_L$ since the key input is the ``local dispersive bound'' of Lemma 2.2 in~\cite{BDH}. It is easy to see that such a bound still holds for
$s\lsim L$ because the Schrodinger evolution is now periodic with period $2\pi L$.  Since we take $L \geq \lambda^{-100} \gg \lambda^{-2}$, this does not affect the proof.

Using Proposition~\ref{prp:zero-bd} we may now transfer bounds from the square-root resolvent of the Laplacian to bounds
on the square-root resolvent of $H_L$.  We always take $z\in\bbH$, so one can define the square root by taking a branch cut down the negative imaginary axis (really we can define $R^{1/2}(z) =\chi(z)$ by applying the functional calculus to any root $\chi(z)$ such that $\chi(z)^2 = \frac{1}{z}$ for $z\in\bbH$).
For the sake of simplicity we record the bound using the stochastic domination notation, but strictly speaking this gives a weaker tail bound than what follows from Proposition~\ref{prp:zero-bd}.

\begin{lemma}
\label{lem:two-to-p-transfer}
Let $X$ be any norm on the vector space of functions of $\bbZ^d_L$.
Then for any $z = E+i\eta$ with $\eta>\lambda^{100}$ we have
\begin{align*}
\|R^{1 /2}(z)\|_{2\to X} \ldom (\lambda \eta^{-1/ 2}+1)\|(\Delta_L-z)^{-1 /2}(z)\|_{2\to X}.
\end{align*}
\end{lemma}
\begin{proof}
By writing
\begin{align*}
	R^{1 /2}(z)=(\Delta_L-z)^{-1 /2}(\Delta_L-z)^{1 /2}R^{1 /2}(z),
\end{align*}
it suffices to show that
\begin{align*}
\|(\Delta_L-z)^{1 /2}R^{1/2}\|_{2\to 2}\prec \lambda \eta^{-1/ 2} + 1.
\end{align*}
To this end, choose a smooth partition of unity $1 = \sum_{k=0}^\infty \chi_k$ where $\chi_0$ and $\chi_j$ are positive bump functions with  supports satisfying
\begin{align*}
\supp \chi_k \subset \begin{cases}
	\{t\mid |t-E|\leq\eta\} &k=0\\
	\{t \mid 2^{k-2}\eta\leq |t-E|\leq 2^{k+1}\eta\} &\text{otherwise}.
\end{cases}
\end{align*}
Moreover, we choose the partition so that the functions $\chi_j$ with $j\geq 1$ are translations and dilations of $\chi_1$ so that Proposition~\ref{prp:zero-bd} can be applied to all $\chi_j$ uniformly.
Note that this implies
\begin{align*}
	\Id=\sum_{k=0}^N\chi_k(\Delta_L)=\sum_{k=0}^N\chi_k(H_L),
\end{align*}
for $N := \log_2 (W/\eta)$ where $W$ is the diameter of the spectrum of $H_L$.
Since $W\ldom \lambda^{-10}$ and $\eta > \lambda^{100}$ we have
$N\ldom 1$.

Using this decomposition we may write,
\begin{align*}
	(\Delta_L-z)^{1 /2}R^{1 /2}(z)=\sum_{k,\ell}(\Delta_L-z)^{1 /2}\chi_k(\Delta_L)\chi_\ell(H_L)R^{1 /2},
\end{align*}
where we have abbreviated $R=R(z)$.
Now observe that for each $k,\ell$ we have
\begin{align*}
&\|(\Delta_L-z)^{1 /2}\chi_k(\Delta_L)\chi_\ell(H_L)R^{1 /2}\|_{2\to 2} \\
& \quad\leq\|(\Delta_L-z)^{1/2}\chi^{1 /2}_k(\Delta_L)\|_{2\to 2}\|\chi^{1 /2}_k(\Delta_L)\chi^{1 /2}_\ell(H_L)\|_{2\to 2}
\|\chi_{\ell}^{1 /2}(H_L)R\|_{2\to 2}\\
&\quad\lsim 2^{k / 2}\eta^{1 /2} \cdot 2^{-\ell / 2}\eta^{-1 / 2}\cdot \|\chi^{1 /2}_k(\Delta_L)\chi^{1 /2}_\ell(H_L)\|_{2\to 2}.
\end{align*}
If $|k-\ell|<3$ then this is bounded by an absolute constant.

Otherwise if $|k-\ell|\geq 3$ then
the support conditions on the $\chi_k$ guarantee that
\begin{align*}
\chi_k^{1 / 2}	(\Delta_L)\chi_{\ell}^{1 /2}(H_L)=-\chi^{1/ 2}_{k}(\Delta_L)\left( \chi^{1 /2}_\ell(\Delta_L)-\chi^{1 /2}_{\ell}(H_L)\right) \\
=\left( \chi^{1 /2}_k(\Delta_L)-\chi^{1 /2}_{k}(H_L)\right) \chi^{1 /2}_{\ell}(H_L).
\end{align*}

By Proposition~\ref{prp:zero-bd}, we have for $K\geq \sqrt{|\log\lambda|}$ the bound
\begin{align*}
\|\chi^{1 /2}_k(\Delta_L)-\chi^{1 /2}_{k}(H_L)\|_{2\to 2}\leq CK|\log\lambda|^22^{-k /2}\eta^{-1 /2},
\end{align*}
for all $k$ with probability at least $1-e^{-cK^2}$. Therefore, for $|k-\ell|\geq 3$ and $K$ in the range given above we have
\begin{align*}
	\|\chi_k^{1 / 2}(\Delta_L)\chi_{\ell}^{1 /2}(H_L)\|_{2\to 2}\leq CK\lambda|\log\lambda|^2 2^{-\max(k,\ell) / 2}\eta^{-1 /2}.
\end{align*}
with probability at least $1-e^{-cK^2}$. Combining these estimates, we conclude that for each $K\geq \sqrt{|\log\lambda|}$
\begin{align*}
\|(\Delta_L-z)^{1 /2}R^{1/2}\|_{2\to 2}&\leq
\sum_{\substack{k,\ell\in[N] \\ |k-\ell|<3}} C +
CK\lambda|\log\lambda|^2 \eta^{-1 /2}\sum_{\substack{k,\ell\in[N]\\|k-\ell|\geq 3}}2^{k/ 2-\ell / 2}2^{-\max(k,\ell) / 2}\\
&\leq CK\lambda |\log \lambda|^2\eta^{-1/2}N + CN,
\end{align*}
with probability at least $1-e^{-cK^2}$.
Using $N\ldom 1$ and the definition of stochastic domination this implies
\begin{align*}
\|(\Delta_L-z)^{1 /2}R^{1/2}\|_{2\to 2}\prec \lambda \eta^{-1 /2} + 1.
\end{align*}
\end{proof}

\subsection{Resolvent bounds for $\Delta_L$}
\label{sec:restriction}
Next we need a slight generalization of the Tomas-Stein restriction estimate in~\cite{tomas1975restriction} (see also~\cite[Section 11.1]{MuscaluSchlag} for a detailed exposition) to the lattice setting.  This was done in~\cite{kachkovskiui2013stein}, and we also include an abbreviated proof below for the sake of completeness.
\begin{proposition}\label{pr:decayToExtension}
Let $\sigma$ be the surface measure of a hyper-surface inside $\bbT^{d}$ satisfying
\begin{align}\label{eq:sigmaHatPtwise}
	|\hat{\sigma}(x)|\leq C\left<x \right>^{-k},
\end{align}
for some $k>0$. Then
\begin{align}\label{eq:TSExtension}
\sup_{f\in\ell^2(\sigma)}\left\|\int_{\bbT^{d}}e^{ix\xi}f(\xi)\,\sigma(d\xi)\right\|_{\ell^p(\Z^{d})}\leq C\|f\|_{\ell^2(\sigma)}
\end{align}
for $p>2(k+1)/k$.
The constant $C$ in \eqref{eq:TSExtension} depends only on the constant in \eqref{eq:sigmaHatPtwise} and the constant
\begin{equation}
\label{eq:ballgrowth}
\sup_{\xi\in\bbT^d, r>0}r^{1-d}\sigma(B(\xi,r)).
\end{equation}
\end{proposition}
\begin{proof}
By a standard duality argument, it suffices to show that
\begin{align*}
	\|f*\hat{\sigma}_E\|_{\ell^{p}(\Z^{d})}\leq C \|f\|_{\ell^{p'}(\Z^{d})}
\end{align*}
for all $f\in \ell^{p'}(\Z^{d})$, where $p'$ is the H\"{o}lder conjugate of $p$.\par
Let $1=\chi_0+\sum_{j=1}^{\infty}\chi_j$ be a Littlewood-Payley partition of unity, i.e., for $j>0 $, $\chi_j(x)=\chi(2^{-j}x)$ for some $\chi\in C^\infty_c(\Real^{d})$.
Restricting to $\Z^{d}$, we may then write
\begin{align*}
f*\hat{\sigma}=f*\chi_0\hat{\sigma}+\sum_{j\geq 1}f*\chi_j \hat{\sigma}.
\end{align*}
From Young's inequality and \eqref{eq:sigmaHatPtwise}, we have that
\begin{align*}
\|f*\chi_j\hat{\sigma}\|_{\ell^\infty(\Z^{d})}\leq C 2^{-j k}\|f\|_{\ell^1(\Z^{d})}.
\end{align*}
On the other hand, we have that
\begin{align*}
\|f*\chi_j \hat{\sigma}\|_{\ell^{2}(\Z^{d})}
&= \|f\|_{\ell^{2}(\Z^{d})}\cdot\|\hat{\chi}_j*\sigma\|_{L^\infty(\bbT^{d})} \\
&\lsim 2^j \|f\|_{\ell^{2}(\Z^{d})}.
\end{align*}
To get the bound $\|\hat{\chi}_j \ast \sigma\|_{L^\infty(\bbT^d)}\lsim 2^j$, note that $\hat{\chi}_j$
is $L^1$-normalized and localized to a ball of radius $2^{-j}$.  More precisely,
$\hat{\chi_j}$ satisfies
\[
|\hat{\chi_j}(\xi)| \lsim 2^{jd} |2^j d(\xi,0) + 1|^{-100d},
\]
for all $\xi\in\bbT^d$ so the bound $\|\hat{\chi_j}\ast \sigma\|_{L^\infty} \lsim 2^j$ follows from~\eqref{eq:ballgrowth}.

Interpolating yields
\begin{align*}
	\|f*\chi_j \hat{\sigma}\|_{\ell^{p}(\Z^{d})}\leq C 2^{j(-k\gamma +(1-\gamma))}\|f\|_{\ell^{p'}(\Z^{d})},
\end{align*}
when $p=\frac{2}{1-\gamma}$. Thus,
\begin{align*}
\|f*\hat{\sigma}\|_{\ell^{p}(\Z^{d})}&\leq C\sum_{j\geq 0}2^{j(-k\gamma +(1-\gamma))}\|f\|_{\ell^{p'}(\Z^{d})}\\
&\leq C \|f\|_{\ell^{p'}(\Z^{d})},
\end{align*}
when $-k\gamma+(1-\gamma)<0$, or equivalently $p>2(k+1)/k$, as desired.
\end{proof}

Next we use Proposition~\ref{pr:decayToExtension} to obtain estimates for spectral projections by
integrating over level sets of the dispersion relation.

\begin{lemma}\label{lm:spectralProj}
Let $\epsilon>0$, $d\ge 2$ and define
\begin{align*}
	\Pi_{E,\alpha}(\cdot) :=\One_{[E-\alpha,E+\alpha]}(\cdot),
\end{align*}
for any $E\in \mathbb{R}$, $\delta>0$. Then for any $E_0\in\mathbb{R}$ with $d(E_0,\Sigma_d)>\epsilon$ and any $\alpha<\epsilon$
\begin{align*}
	\|\Pi_{E_0,\delta}(\Delta_L)\|_{2\to p}\lsim  \sqrt{\alpha}.
\end{align*}
\end{lemma}

\begin{proof}
We prove the analogous bound for $\Delta_{\Z^d}$, which then implies the statement by Lemma~\ref{lem:exp-decay}. Fix $\epsilon>0$. First, for $d=3$, let $\sigma_E$ be the surface measure on $\omega^{-1}(\{E\})$.
From Section 4.3 of \cite{CueninSchippa}, there exists
a $C>0$ such that for any $E$ with $d(E,\Sigma_3)>\eps$
\begin{align*}
|\hat{\sigma}_E(x)|\leq C\left<x \right>^{-3 /4}.
\end{align*}

Moreover, by continuity and compactness, the quantity
\begin{align*}
\sup_{r>0,\xi\in \bbT^{d}}\frac{\sigma_E(B(\xi,r))}{r^{d-1}}
\end{align*}
is bounded uniformly for $E$ in this range as well so that by Proposition \ref{pr:decayToExtension} we have
\begin{align}\label{eq:extensionEst}
\sup_{f\in\ell^2(\sigma)}\left\|\int_{\bbT^{d}}e^{ix\xi}f(\xi)\,\sigma(d\xi)\right\|_{\ell^p(\Z^{d})}\leq C\|f\|_{\ell^2(\sigma)},
\end{align}
for any $p>14/ 3$.  Using the coarea formula, we write
\begin{align*}
(\Pi_{E_0,\delta}(\Delta)f)(x)=
\int_{|\omega(\xi)-E_0|<\delta}e^{ix\xi}\hat{f}(\xi)\:d\xi
=\int_{E_0-\delta}^{E_0+\delta}\int e^{ix\xi}\hat{f}(\xi)|\nabla \omega(\xi)|^{-1} \:\sigma_E(d\xi) dE,
\end{align*}
where $|\nabla\omega|$ is uniformly bounded below on the domain of integration.
Thus, from Minkowski's inequality and \eqref{eq:extensionEst}, we have that
\begin{align*}
\|\Pi_{E_0,\delta}(\Delta)f\|_{\ell^p(\Z^{3})}&\leq
\int_{E_0-\delta}^{E_0+\delta}\left\|\int e^{ix\xi}\hat{f}(\xi)|\nabla \omega(\xi)|^{-1}\:\sigma_E(d\xi)\right\|_{\ell^{p}(\Z^{3})}\:dE\\
&\leq C\int_{E_0-\delta}^{E_0+\delta}\|\hat{f}\|_{L^2(\sigma_E)}\:dE\\
&\leq C\sqrt{\delta} \|f\|_{\ell^2(\Z^{3})},
\end{align*}
which establishes the result for $d=3$.

For $d=2$, the proof is the same except that one instead has the estimate
\begin{align*}
|\hat{\sigma}_E(x)|\leq C\left<x \right>^{-1 /2},
\end{align*}
yielding the same result for $p>6$.
This follows from elementary stationary phase because the curvature of $\omega^{-1}(\{E\})$ is uniformly bounded away from $0$ for all $E$ with $d(E,\Sigma_2)>\eps$.  Indeed, this same bound holds in
any $d\geq 2$ because there is always at least one direction of nonvanishing Gaussian curvature.  In particular, we can take $p_d=6$ in $d=2$ and $d=4$

For $d>4$, we instead rewrite  $\Pi_{E_0,\delta}$ in terms of the unitary propagator:
\begin{align*}
\Pi_{E_0,\delta}(\Delta_{\Z^{d}})=\int_{\R}\hat{\One}_{[E_0-\delta,E_0+\delta]}(t)e^{it\Delta_{\Z^{d}}}\:dt
\end{align*}
so that by Cauchy-Schwarz
\begin{align*}
\|\Pi_{E_0,\delta}(\Delta_{\Z^{d}})\|_{2\to p} \lsim
\sqrt{\delta}\|e^{it\Delta_{\Z^{d}}}\|_{L^2(\R\to \ell^p(\Z^{d}))}\\
\end{align*}
Using the Strichartz estimate for $\Delta_{\Z^{d}}$ \cite[Thm 1]{stefanov2005asymptotic},
\begin{align*}
\|e^{it\Delta_{\Z^{d}}}\|_{L^2_t\ell^p_x}\lsim 1,
\end{align*}
valid for $d>3$ and $p=\frac{2d}{d-3}$, completes the proof.
\end{proof}

Finally we decompose the square root of the resolvent $R^{1/2}$ using spectral projections and complete the proof of Proposition~\ref{prp:apriori}.

\begin{proposition}
\label{prp:R_0Bounds}
Let $d\geq 2$ and $\epsilon>0$ and $p>p_d$.
For any $z=E+i\eta$ with $d(E,\Sigma_d)>\epsilon$
\begin{align}\label{eq:freeResolvent2Top}
\|(\Delta_L-z)^{-1 /2}\|_{2\to p}\lsim |\log \eta| + 1
\end{align}
and
\begin{align}\label{eq:freeResolventpTo2}
\|(\Delta_L-z)^{-1 /2}\|_{p'\to 2}\lsim |\log \eta| + 1.
\end{align}
The implicit constant may depend on $p$.
\end{proposition}

\begin{proof}
For $k>0$, let $\Pi_k$ be the indicator function of
\begin{align*}
\{E' \mid 2^{k-1}\eta<|E-E'|\leq 2^{k}\eta\} \cap \{E' \mid d(E',\Sigma_d)>\eps\},
\end{align*}
let $\Pi_0$ be the indicator function of
\begin{align*}
\{E' \mid \eta|E-E'|\leq \eta\} \cap \{E' \mid d(E',\Sigma_d)>\eps\},
\end{align*}
and let $\Pi^c$ be the indicator function of
\begin{align*}
\{E' \mid d(E',\Sigma_d)\leq \eps\}.
\end{align*}
We may therefore write
\begin{align*}
(\Delta_L - z)^{-1 /2}=(\Delta_L - E - i\eta)^{-1 /2}
\left( \Pi^c(\Delta)+ \sum_{k=0}\Pi_k(\Delta)\right) ,
\end{align*}
where the sum over $k$ contains at most $\log \eta$ terms.

Since $p\geq 2$, we have that
\begin{align*}
\|(\Delta_L - z)^{-1 /2}\Pi^c(\Delta_L)\|_{p \to 2}&\leq
\|(\Delta_L - z)^{-1 /2}\Pi^c(\Delta_L)\|_{2 \to 2}\\
&=\sup_{d(E',\Sigma_d)\leq \eps}\left( |E-E'|^2+\eta^2 \right)^{-1 /4},
\end{align*}
which is bounded in terms of $\eps$, independently of $\eta$, by the assumption on $E$.

Otherwise, we use Lemma \ref{lm:spectralProj} to see that
\begin{align*}
\|(\Delta_L-z)^{-1 /2}\sum_{k\geq 0}\Pi_k(\Delta_L)\|_{2\to p}&\leq
\sum_{k\geq 0}\|(\Delta-z)^{-1 /2}\Pi_k(\Delta)\|_{2\to 2}\|\Pi_k(\Delta)\|_{2\to p}\\
&\leq C \sum_{k\geq 0}(2^{-k /2}\eta^{-1 /2})\cdot (2^{k /2} \eta^{1 /2} )\\
&\leq C|\log\eta|,
\end{align*}
which establishes \eqref{eq:freeResolvent2Top}.
The estimate \eqref{eq:freeResolventpTo2} now follows by duality since $R^*(z)=R(\overline{z})$.
\end{proof}

\begin{proof}[Proof of Proposition~\ref{prp:apriori}]
We simply apply Lemma~\ref{lem:two-to-p-transfer} and Proposition~\ref{prp:R_0Bounds} successively:
\begin{align*}
\|R(z)\|_{p\to q}&\leq \|R^{1 /2}(z)\|_{p\to 2}\|R^{1 /2}(z)\|_{2\to q}\\
&\prec (\lambda\eta^{-1/2}+1)^{2}\|(\Delta_L-z)^{-1/2}\|_{p \to 2}\|(\Delta_L-z)^{-1/2}\|_{2\to q}\\
&\prec \lambda^2\eta^{-1} + 1 ,
\end{align*}
as desired.  To bound $\|R\|_{p\to 2}$ we use instead
\begin{align*}
\|R(z)\|_{p\to 2}
&\leq \|R^{1/2}(z)\|_{p\to 2} \|R^{1/2}(z)\|_{2\to 2} \\
&\ldom (\lambda \eta^{-1/2}+1)\|(\Delta_L-z)^{-1/2}(z)\|_{p\to 2} \|R^{1/2}(z)\|_{2\to 2} \\
&\ldom \eta^{-1/2}(\lambda \eta^{-1/2}+1),
\end{align*}
where we used that $\|R^{1/2}(z)\|_{2\to 2} \leq \eta^{-1/2}$, which holds deterministically.
\end{proof}

\section{Analysis of the self-consistent equation: proof of the local law}
\label{sec:sce}
In this section, we follow the approach outlined in Section~\ref{sec:sce-intro}
and show $R(z)$ is entrywise close to $M(z)$, defined as the solution to the self-consistent equation
\begin{align}
\label{eq:def-of-M}
M(z) = (\Delta_L- (z+\lambda^2M_{00}(z)))^{-1}.
\end{align}
Note that for any $z\in \mathbb{H}$, $M(z)$ exists and is unique by Proposition \ref{pr:thetaAPriori}, and that by Proposition~\ref{prp:approxSCE}, $\mathbb{E} R$ approximately solves a similar equation.

\begin{theorem}[Local law]
\label{thm:local-law}
Fix $\epsilon,\delta>0$ and $d\geq 2$, and let $z = E+i\eta$ with $d(E,\Sigma_d)\geq \epsilon$.
Then
\begin{itemize}
\item \textbf{In any} $\mathbf{d \geq 2}$,
and $\eta > \lambda^{2+2/9-\delta}$,
\begin{align*}
\sup_{x,y\in \Z^{d}_{L}}|\Expec R_{xy} - M_{xy}| &\lsim
\lambda^{\frac12-\delta} (\lambda^2\eta^{-1}+1)^{9/4} \\
\sup_{x,y\in \bbZ^{d}_{L}}|R_{xy} - \Expec R_{xy}| &\ldom \lambda^{\frac12} (\lambda^2\eta^{-1}+1)^{7/4}.
\end{align*}
\item \textbf{In} $\mathbf{d=3}$ and $\eta > \lambda^{2+6/19-\delta}$,
\begin{align*}
\sup_{x,y\in \Z^{d}_{L}}|\Expec R_{xy} - M_{xy}| &\lsim \lambda^{\frac34-\delta}
(\lambda^2\eta^{-1}+1)^{19/8} \\
\sup_{x,y\in \Z^{d}_{L}}|R_{xy}-\Expec R_{xy}| &\ldom \lambda^{\frac34}
(\lambda^2\eta^{-1}+1)^{15/8}.
\end{align*}
\item \textbf{In} $\mathbf{d\geq 7}$ and $\eta > \lambda^{2+\frac{2d-12}{3d-12} - \delta}$,
\begin{align*}
\sup_{x,y\in \Z^{d}_{L}}|\Expec R_{xy} - M_{xy}| &\lsim \lambda^{1-\delta} (\lambda^2 \eta^{-1}+1)^{1/2}.
\\
\sup_{x,y\in \Z^{d}_{L}}|R_{xy}-\Expec R_{xy}| &\ldom \lambda.
\end{align*}
\end{itemize}
All implicit constants above depend only on $\delta, \epsilon$, and $d$, and $M=M(z)$
is defined by~\eqref{eq:def-of-M}.
\end{theorem}

In the course of the proof it will be clear that an improvement to the Tomas-Stein exponent $p_d$ from the previous section would lead to an improvement in the exponent of $\lambda$ in $d>3$.
Moreover, for simplicity we treat $d\in\{2,4,5,6\}$ in a unified way, although in principle Proposition~\ref{prp:apriori} allows for an improvement in $d\in\{5,6\}$.

The fluctuations of $R_{xx}$ are at least of order $\lambda$ as can be seen by varying $V_{xx}$.
So the bound on $|R_{xx}-\mathbb{E}R_{xx}|$ that we obtain in $d\geq 7$ is sharp to leading order.

\subsection{The key ingredients}

We use two principles in the proof of Theorem~\ref{thm:local-law} (which are also applied
to understand $\Expec RAR^*$):
\begin{enumerate}
\item $\Expec R$ solves \eqref{eq:def-of-M} up to error terms which are
controlled by the entrywise fluctuations of $R$.
\item The entrywise fluctuations of $R$ are controlled by $\|R\|_{1\to 4}$.
\end{enumerate}
The first principle combined with the \textit{stability} of \eqref{eq:def-of-M}
implies that to bound $\|\mathbb{E}R-M\|_{1\to \infty}$, one just needs to control the fluctuations of $R$.
This idea is by now standard in random matrix theory. The second principle is specific to random
Schrodinger operators and comes from the fact that the randomness only appears on the diagonal.

Below, we will frequently make use of the following special case of the resolvent identity, often known as the \emph{Ward identity}:
\begin{align}\label{eq:Ward}
R(z) R^*(z)=\frac{\Impt R(z)}{\Impt z}\\
\nonumber\Impt R=\frac{R-R^*}{2i},
\end{align}
valid for $R$ the resolvent of any self-adjoint operator and $\Impt z > 0$.  In particular, taking a diagonal entry we have for any $x\in\bbZ^d_L$ and $z=E+i\eta$
\begin{equation}
\label{eq:explicit-Ward}
\sum_{y\in \mathbb{Z}^d_L} |R(E+i\eta)_{xy}|^2 = \eta^{-1} \Impt R(z)_{xx}.
\end{equation}
This allows us to transfer of $\ell^2$-based quantities to entrywise ones.
For instance, taking a supremum over $x$ in~\eqref{eq:explicit-Ward} and taking a square root we have
\begin{align}\label{eq:2toInftyWard}
\|R\|_{1\to 2}\leq (|\Impt z|)^{- 1/2}\|R\|_{1\to \infty}^{1 /2}.
\end{align}
Here and below we sometimes use the $1\to \infty$ operator norm as a convenient shorthand for the entrywise supremum.

First we state a calculation showing $\mathbb{E} R_{xy} - M_{xy}$ is bounded by the fluctuations of $R_{xy}$.
\begin{lemma}[Fluctuations control expectations]
\label{lem:fluct-to-er}
Let $\delta,\epsilon>0$, $d\geq 2$, and $z$ be as in Theorem ~\ref{thm:local-law}.
If $ \lambda\ll 1$ and
\begin{align}\label{eq:ERMInitial}
   |\Expec R_{00}-M_{00}|\ll 1
\end{align}
then
\begin{align*}
\|\Expec R - M\|_{1\to \infty}
\lsim(\lambda^2\eta^{-1})^{1 /2} \left(  \Expec \|R-\Expec R\|_{1\to \infty}^{2}\right) ^{1/ 2}
\left( \Expec\|R\|_{1\to \infty} \right)^{1 /2}.
\end{align*}

\end{lemma}
\begin{proof}
First we use the approximate self-consistent equation to estimate $\mathbb{E}R_{xy}-\tilde{M}_{xy}$. Recall from Proposition~\ref{prp:approxSCE} that
\begin{equation}\tag*{(SCE)}
\label{eq:apxsce}
\begin{split}
&\mathbb{E}R = \tilde{M}(\Id+\mcal{E}_{\rm loc}) \\
&\mathfrak{E}_{\rm loc} := \lambda^2\mathbb{E}\mathcal{D}[R-\Expec R]R.
\end{split}
\end{equation}
where the diagonalizing superoperator $\mcal{D}$ was defined in~\eqref{eq:D-def} and
\begin{align*}
\tilde{M}(z)=(\Delta_L-(z+\lambda^{2}\Expec R_{00}(z)))^{-1}.
\end{align*}

For any $x,y\in \Z^{d}_{L}$, we may write using \ref{eq:apxsce}
\begin{align*}
|\Expec R_{xy} - \tilde{M}_{xy}|
&= \lambda^2 |\braket{x|\tilde{M} \Expec (\mcal{D}[R-\Expec R])R|y}| \\
&\leq \lambda^2\|\tilde{M}\ket{x}\|_{2}(\Expec \|\mcal{D}[R-\Expec R]\|_{2\to 2}^2)^{1/2}
(\Expec \|R\|_{1\to 2}^2)^{1/2} \\
&\lsim \lambda\eta^{-1 / 2} \left(  \Expec \|R-\Expec R\|_{1\to \infty}^{2}\right) ^{1/ 2}
\left( \Expec\|R\|_{1\to \infty} \right)^{1 /2}.
\end{align*}
To pass to the last line we used the estimates $\mathbb{E}\|R\|_{1\to 2}^2\lsim \eta^{-1}\mathbb{E}\|R\|_{1\to\infty}$
and $\|\tilde{M}\|_{1\to 2}\lsim \lambda^{-1}$. The first follows directly by
\eqref{eq:2toInftyWard}. For the second we first apply the Ward identity \eqref{eq:explicit-Ward}
to $\tilde{M}$ to estimate
\begin{align*}
\|\tilde{M}\ket{x}\|_{2}^2
& = \frac{\Impt \tilde{M}_{xx}}{\eta + \lambda^2\Impt\mathbb{E}R_{00}(z)}.
\end{align*}
Now note that, since $d(E,\Sigma_d)>\eps$,  Lemma \ref{lem:dos-regularity} and \eqref{eq:ERMInitial} imply $\Impt(\mathbb{E}R_{00})\gtrsim 1$ for $\lambda$
small enough.  Moreover, $\text{Re}(\mathbb{E}R_{00})\lsim 1$, so that
taking $\lambda$ sufficiently small we can apply Lemma \ref{lem:dos-regularity} to $\tilde{M}$ and obtain that $\Impt{\tilde{M}_{xx}}\lsim 1$.

Now we estimate $|\tilde{M}_{00} - M_{00}|$. The function $(\Delta_{L}-z)^{-1}_{00}$ is  $C^1$ in the $z$ variable away from $\Sigma_d$
(see Lemma~\ref{lem:dos-regularity}), so
\begin{align*}
    |\tilde{M}_{00} - M_{00}|
&= |(\Delta_{L} - (z+\lambda^2\Expec R_{00}))^{-1}_{00}
- (\Delta_L - (z+\lambda^2 M_{00}))^{-1}_{00}| \\
&\lsim \lambda^2 |\Expec R_{00} - M_{00}|.
\end{align*}
Putting the above two bounds together gives us the claim for diagonal entries:
\begin{align*}
|\Expec  R_{00} - M_{00}|
&\leq |\Expec R_{00} - \tilde{M}_{00}| + |\tilde{M}_{00}-M_{00}| \\
&\lsim \lambda\eta^{-1 / 2} \left(  \Expec \|R-\Expec R\|_{1\to \infty}^{2}\right) ^{1/ 2}
\left( \Expec\|R\|_{1\to \infty} \right)^{1 /2}
 + \lambda^2 |\Expec R_{00} - M_{00}|,
\end{align*}
which can be rearranged to give
\[
|\Expec R_{00} - M_{00}| \lsim
\lambda\eta^{-1 / 2} \left(  \Expec \|R-\Expec R\|_{1\to \infty}^{2}\right) ^{1/ 2}
\left( \Expec\|R\|_{1\to \infty} \right)^{1 /2},
\]
provided that $\lambda$ is sufficiently small, depending on $\epsilon$, $\delta$ and $d$.

For the off-diagonal entries we apply the resolvent identity, the Ward identity,
and the above to estimate
\begin{align*}
\|\tilde{M}-M\|_{1\to\infty}
& =\lambda^{2}\|\tilde{M}(\Expec R_{00}-M_{00})M\|_{1\to\infty}\\
& \lsim |\mathbb{E}R_{00}-M_{00}|\\
& \lsim \lambda\eta^{-1 / 2} \left(  \Expec \|R-\Expec R\|_{1\to \infty}^{2}\right) ^{1/ 2}
\left( \Expec\|R\|_{1\to \infty} \right)^{1 /2},
\end{align*}
where we used that $\|M\|_{1\to 2}, \|\tilde{M}\|_{1\to 2}\lsim \lambda^{-1}.$
Writing
\begin{align*}
\|\Expec R-M\|_{1\to \infty}\leq C\lambda\eta^{-1 / 2}
\left(  \Expec \|R-\Expec R\|_{1\to \infty}^{2}\right) ^{1/ 2}
\left( \Expec\|R\|_{1\to \infty} \right)^{1 /2}
+\|M-\tilde{M}\|_{1\to \infty},
\end{align*}
and applying the previous inequality thus completes the proof.
\end{proof}

Next, we show that $\|R\|_{1\to 4}$ controls the fluctuations of the entries of $R$.
We use the following variant of the Poincar{\'e} inequality for Gaussian random
variables.
\begin{lemma}[Proposition 5.4.2 of~\cite{bakry2014analysis}]
\label{lem:lsi-poincare}
Let $f:\Real^N\to\bbC$ be a function of independent standard Gaussian random variables.  Then
for $k\geq 1$
\begin{equation}
\label{eq:lsi-poincare}
\Big(\Expec |f|^{2k}\Big)^{1/k} \leq \Expec |f|^2
+ C(2k-2)\Big(\Expec |\nabla f|^{2k}\Big)^{1/k}.
\end{equation}
and in particular
\begin{equation}
\label{eq:lsi-moments}
(\Expec |f-\mathbb{E}f|^{2k})^{1/2k} \leq C \sqrt{k} (\Expec |\nabla f|^{2k})^{1/2k}.
\end{equation}
\end{lemma}
To simplify matters, we observe that~\eqref{eq:lsi-moments} implies the
stochastic domination relationship
\[
f \ldom |\nabla f|.
\]

Applied to $R$ and $RAR^*$ (which we use later in Section~\ref{sec:diffusion-profile}),
we obtain the following concentration inequalities.
\begin{lemma}
\label{lem:R-lsi}
For any $x,y\in\bbZ^d_L$ and $z=E+i\eta$ with $\eta>0$ we have
\begin{equation}
\label{eq:xRy-conc}
|R_{xy}-\mathbb{E}R_{xy}|\prec \lambda \|R\|_{1\to 4}^2
\end{equation}
and for any $A\in\mcal{B}(\ell^2(\bbZ^d_L))$, we have
\begin{equation}
\label{eq:xRARy-conc}
|(RAR^{\ast})_{xy} - \mathbb{E}(RAR^{\ast})_{xy}|\prec \lambda  \|A\|_{2\to 2} \|R\|_{1\to 4}
\|R\|_{1\to 2}
\|R\|_{2\to 4}
\end{equation}
\end{lemma}
\begin{proof}
First we prove~\eqref{eq:xRy-conc}.  Write the potential as
\[
V = \sum_{x\in\bbZ^d_L} g_x \ket{x}\bra{x}.
\]
Viewing $R_{xy}$ as a function of the Gaussians $\{g_x\}_{x\in\bbZ^d_L}$, the following formula follows from the resolvent identity:
\[
\partial_{g_w} R_{xy} = -R_{xw}R_{wy}.
\]
Therefore, writing $\nabla_V$ for the gradient with respect to all of the entries of $V$,
we have
\begin{align*}
|\nabla_{V} R_{xy}|^2
& = \lambda^2 \sum_{\alpha\in\bbZ^d_L} |R_{x\alpha}R_{\alpha y}|^2\\
& \leq \lambda^2 ||R||_{1\to 4}^4,
\end{align*}
where the final inequality follows from Cauchy-Schwarz.
Now the claim follows by~\eqref{eq:lsi-poincare} and the
definition of our stochastic domination notation.

The proof of ~\eqref{eq:xRARy-conc} is similar. We compute
\[
\partial_{g_w} (RAR^*)_{xy} = -R_{xw}(RAR^*)_{wy} - (RAR^*)_{xw}R_{wy},
\]
and therefore
\begin{align*}
|\nabla_{V} (RAR^*)_{xy}|^2
& \leq 2\lambda^2 \sum_{\alpha\in \mathbb{Z}^d_L}|R_{x\alpha}(RAR^*)_{\alpha y}|^2 + 2\lambda^2 \sum_{\alpha\in \mathbb{Z}^d_L}|(RAR^{\ast})_{x\alpha}R^{\ast}_{\alpha y}|^2 \\
& \leq 2\lambda^2 (||R||_{1\to 4}^2 ||RAR^{\ast}||_{1\to 4}^2 +  ||R^{\ast}||_{1\to 4}^2 ||RAR^{\ast}||_{1\to 4}^2)\\
& \leq 4\lambda^2||R||_{1\to 4}^2 ||R||_{1\to 2}^2 ||A||_{2\to 2}^2 ||R||_{2\to 4}^2.
\end{align*}
\end{proof}

\subsection{Proof of Theorem~\ref{thm:local-law} for $\mathbf{d\geq 2}$}
For concreteness we write out the argument for $d=3$, the argument for all other $d\geq 2$ differing only in arithmetic.
In this dimension Proposition \ref{prp:apriori} gives the bound
\[
\|R\|_{1\to \frac{14}{3}} \ldom \lambda^2\eta^{-1} + 1.
\]
In dimension $d=2,4,5,6$ we instead use the exponent\footnote{For $d\in\{4,5,6\}$ this exponent is not sharp.} $6$ in place of $\frac{14}{3}$.

\begin{lemma}
\label{lem:five-plus-twelve}
Let $d=3$, and $\epsilon>0$. Then for $z= E+i\eta$ with $d(E,\Sigma_d)\geq \epsilon$ and $\eta\geq \lambda^{10}$, we have
\[
\|R\|_{1\to 4} \ldom \lambda^{-1/8} (\lambda^2 \eta^{-1}+1)^{15/16}
(\|R\|_{1\to \infty})^{1/16}.
\]
\end{lemma}
\begin{proof}
Using interpolation and the Ward identity $\|R\|_{1\to 2}\leq \eta^{-1/2} \|R\|_{1\to\infty}^{1/2}$ we have,
\begin{align*}
\|R\|_{1\to 4}
&\leq \|R\|_{1\to 2}^{1/8} \|R\|_{1\to \frac{14}{3}}^{7/8} \\
&\leq \eta^{-1/16} \|R\|_{1\to \infty}^{1/16} \|R\|_{1\to \frac{14}{3}}^{7/8}.
\end{align*}
The conclusion follows from applying the a priori bound for $\|R\|_{1\to \frac{14}{3}}$ and rearranging
factors of $\lambda$ and $\eta$.
\end{proof}

\begin{proof}[Proof of Theorem~\ref{thm:local-law} in the case $d=3$]
Fix $E$ with $d(E,\Sigma_d)\geq \epsilon$ and let  $a=3 /4$ and $b=15 / 8$ for convenience.
First we show for each $\delta>0$ there exists $C_\delta>0$ such that
\begin{align}
\label{eq:expecBoundab}
\|\Expec R-M\|_{1\to \infty}\leq
C_{\delta}\lambda^{a-\delta}(\lambda^{2} \eta^{-1}+1)^{b+1 /2}
\end{align}
for all $\lambda \ll_{\delta}1$ and $\eta\in [\lambda^{2+6/19-\delta},1]$.
To do this, we prove there exists $c_{\delta}, C_{\delta}>0$ such that if
\begin{align}
\label{eq:bootstrapAssumption}
\|\Expec R - M\|_{1\to \infty}\leq c_{\delta}
\end{align}
for some $\lambda \ll_{\delta}1$ and $\eta\in [\lambda^{2+6/19-\delta},1]$ then in fact
\begin{align}
\label{eq:bootstrapConclusion}
\|\Expec R - M\|_{1\to \infty}\leq
C_{\delta}\lambda^{a-\delta}(\lambda^2\eta^{-1}+1)^{b+1 /2}.
\end{align}
The conclusion then follows by a standard continuity argument.
Indeed, summing the Born series for $R(z)$ shows that \eqref{eq:expecBoundab}, holds for $\eta= 1$ and some $C_{\delta}$ independent of $\lambda$.
Moreover, $\|\mathbb{E}R-M\|_{1\to\infty}$ is  $\lambda^{-10}$-Lipschitz for $\eta\in [\lambda^{2+6/19-\delta},1]$,
and so the above claim implies that \eqref{eq:bootstrapConclusion} continues to hold as long as
$$C_{\delta}\lambda^{a-\delta}(\lambda^2\eta^{-1}+1)^{b+1 /2}\leq c_{\delta}/2,$$
which is satisfied for all $\lambda \ll_{\delta}1$ and $\eta\in [\lambda^{2+6/19-\delta},1]$.

To prove that \eqref{eq:bootstrapAssumption} implies \eqref{eq:bootstrapConclusion}, we note that
if $c_{\delta}$ is sufficiently small we can apply Lemma
\ref{lem:fluct-to-er} to estimate
\begin{align}
\label{eq:FluctuationsBoundExpec}
\|\Expec R-M \|_{1\to \infty}\leq
C(\lambda^2\eta^{-1}+1)^{1/ 2} \left(  \Expec \|R-\Expec R\|_{1\to \infty}^{2}\right) ^{1/ 2}
\left( \Expec\|R\|_{1\to \infty} \right)^{1 /2}.
\end{align}

Now, we use Lemmas \ref{lem:R-lsi} and \ref{lem:five-plus-twelve} to see that
\begin{align}
\label{eq:RBoundsFluctuations}
(\Expec \|R-\Expec R\|_{1\to \infty}^{2})^{1 /2}\leq C_\delta
\lambda^{a-\delta /2}(\lambda^{2}\eta^{-1}+1)^{b}(\Expec\|R\|_{1\to \infty})^{1 /8},
\end{align}
which when inserted into \eqref{eq:FluctuationsBoundExpec} yields
\begin{align}
\label{eq:simpleExpecBound}
\|\Expec R-M \|_{1\to \infty}\leq C_\delta \lambda^{a-\delta /2}(\lambda^{2}\eta^{-1}+1)^{b+1 /2}(\Expec \|R\|_{1\to \infty})^{5/8}.
\end{align}
Furthermore, by \eqref{eq:bootstrapAssumption} and \eqref{eq:RBoundsFluctuations}, we have
\begin{align*}
\Expec \|R\|_{1\to \infty}
&\leq \Expec \|R-\Expec R\|_{1\to \infty} +\|\Expec R-M\|_{1\to \infty}+\|M\|_{1\to\infty}\\
&\leq C_{\delta}\lambda^{\delta}(\Expec\|R\|_{1\to \infty})^{1 /8} + c_{\delta} + C_{\epsilon}
\end{align*}
where we have used Proposition \ref{pr:thetaAPriori} to control $\|M\|_{1\to \infty}$, and that
$\lambda^{a-\delta /2}(\lambda^{2}\eta^{-1}+1)^{b}\leq \lambda^{\delta}$ for $\eta$ and $\lambda$ in the
assumed ranges.
Taking $\lambda$ sufficiently small relative to $C_\delta$ ensures that
\begin{align*}
\Expec \|R\|_{1\to \infty} \leq  C_\eps
\end{align*}
so that by \eqref{eq:simpleExpecBound}
\begin{align*}
\|\Expec R - M\|_{1\to \infty} \leq C_\delta \lambda^{\delta / 2} \lambda^{a-\delta }(\lambda^{2}\eta^{-1}+1)^{b+1 /2},
\end{align*}
which proves \eqref{eq:bootstrapConclusion}, and thus \eqref{eq:expecBoundab}

Finally, to prove that
\begin{align*}
\|R-\Expec R\|_{1\to \infty}\ldom \lambda^{a}(\lambda^2\eta^{-1}+1)^{b}
\end{align*}
we simply note that by Lemmas
\ref{lem:R-lsi} and \ref{lem:five-plus-twelve}
\begin{align*}
\|R-\Expec R\|_{1\to \infty}&\ldom \lambda^{a}(\lambda^{2}\eta^{-1}+1)^{b}\|R\|_{1\to \infty}^{1 /4}\\
&\leq \lambda^{a}(\lambda^{2}\eta^{-1}+1)^{b}
\left(\|R-\Expec R\|_{1\to \infty}^{1 /4}+\|\Expec R-M\|_{1\to \infty}^{1 /4}+C_\eps \right).
\end{align*}
The conclusion now follows by inserting \eqref{eq:expecBoundab} and using the definition of $\ldom$ to rearrange.
\end{proof}

\subsection{Proof of Theorem~\ref{thm:local-law} in $\mathbf{d\geq 7}$}

For $d\geq 7$ the proof is almost identical but with one
modification.  Instead of using Lemma \ref{lem:five-plus-twelve} to bound $\|R\|_{1\to 4}$ we use
the following lemma, which gives a better bound.
\begin{lemma}
\label{lem:R-flucs-highd}
Let $d\geq 7$, $\epsilon,\delta>0$ and $z=E+i\eta$ be as in Theorem \ref{thm:local-law} in the $d\geq 7$ case.
If  $\|\mathbb{E}R-M\|_{1\to\infty}\ll 1$, then
\begin{equation}
\label{eq:bootstrapped-1-to-4}
    \|R\|_{1\to 4} \ldom 1
\end{equation}
and thus
\[
|R_{xy} - \Expec R_{xy}| \ldom \lambda.
\]
\end{lemma}
\begin{proof}
Suppose $\|\mathbb{E}R-R\|_{1\to\infty} \ldom F$ for some $F\geq \lambda $ (a priori we can use for example $F = \lambda^{-10}$).
Combined with Lemma~\ref{lem:fluct-to-er} this implies $|R_{xy}-M_{xy}| \ldom (\lambda \eta^{-1/2}+1)F$.
Thus for any $x\in \mathbb{Z}^d_L$ if we let
$$S_x:=\{y\in \mathbb{Z}^d_L |\ |M_{xy}|\geq (\lambda \eta^{-1/2}+1)F\},$$
we can estimate for any $\beta\in[p_d-2,2]$
\begin{align*}
\sum_{y\in \bbZ^d_L} |R_{xy}|^4
&\leq \sum_{y\in S_{x}} |R_{xy}|^4 + \sum_{y\notin S_{x}} |R_{xy}|^4 \\
&\ldom \|M\|_{1\to 4}^4 + ((\lambda \eta^{-1/2} +1)F)^{2-\beta} \sum_{y\in \bbZ^d_L} |R_{xy}|^{2+\beta} \\
&\ldom \|M\|_{1\to 4}^4 + ((\lambda \eta^{-1/2}+1) F)^{2-\beta} \|R\|_{1\to 2+\beta}^{2+\beta}\\
& \ldom 1 + (\lambda^2 \eta^{-1}+1)^{3+\beta/2} F^{2-\beta}
\end{align*}
To pass to the last line we used Proposition \ref{prp:R_0Bounds} in $d\geq 6$ to bound $\|M\|_{1\to 4}\lsim 1$,
and Proposition \ref{prp:apriori} to estimate $\|R\|_{1\to 2+\beta}^{2+\beta}\ldom (\lambda^2\eta^{-1}+1)^{2+\beta}$.
Taking the maximum over $x\in\bbZ^d_L$ shows
\begin{equation}
\label{eq:temp-1-to-4}
\|R\|_{1\to 4}^4\ldom 1 + (\lambda^2 \eta^{-1}+1)^{3+\beta/2} F^{2-\beta}.
\end{equation}
By Lemma~\ref{lem:R-lsi}, this implies
\begin{align*}
|R_{xy} - \Expec R_{xy}| \ldom \lambda + \lambda (\lambda^2\eta^{-1}+1)^{\frac32 + \beta/4} F^{1-\beta/2}.
\end{align*}
Now setting $\eta = \lambda^{2+\gamma}$ and using $F\gtrsim \lambda$, we obtain $\|\Expec R -R\|_{1\to\infty} \ldom F'$ where
\[
F' = \lambda + (\lambda^{1-\frac\beta2-\gamma(\frac32+\frac\beta4)}) F.
\]
As long as $\gamma \leq  \frac{1-\frac\beta2}{\frac{3}{2}+\frac{\beta}{4}} - \delta$, this can be iterated finitely
many times to conclude that $\|\Expec R - R\|_{1\to\infty} \ldom \lambda$. The fraction is positive precisely when $\beta<2$ which happens for $d\geq 7$. Furthermore, recalling that $p_d = 2d/(d-3)$ and $\beta = p_d-2 = \frac{6}{d-3}$ gives the condition
$$\gamma \leq \frac{2d-12}{3d-12} - \delta.
$$
which was the lower bound we assumed on $\eta$.
Combining this with \eqref{eq:temp-1-to-4} then proves the bound $\|R\|_{1\to 4}\ldom 1$.
\end{proof}

Now Theorem~\ref{thm:local-law} follows for $d\geq 7$ by repeating the proof given above
for $d=3$ but substituting Lemma \ref{eq:bootstrapped-1-to-4}
for Lemma \ref{lem:five-plus-twelve} which gives the bound
$$(\mathbb{E}\|R-\mathbb{E}R\|_{1\to\infty}^2)^{1/2}\leq \lambda^{1-\delta},$$
in place of \eqref{eq:RBoundsFluctuations}. We note that the conditions of Lemma \ref{eq:bootstrapped-1-to-4}
are satisfied since when Lemma
\ref{lem:five-plus-twelve} is used in the proof above, we are under the assumption that $\|\mathbb{E}R-M\|_{1\to\infty}\ll 1$.

\section{Diffusion Profile}
\label{sec:diffusion-profile}
In this section we obtain high probability bounds on quantities of the form
\begin{equation}
\label{eq:second-moment-sum}
\sum_{y\in \Z^d} a(y) |R_{xy}|^2
\end{equation}
for some $x\in\bbZ^d_L$ and $a \in \ell^\infty(\bbZ^d_L)$.  This expression can alternatively be written as
\[
\Expec \braket{x|RA R^*|x}
\]
for a diagonal operator $A$ with entries $a(x)$ on the diagonal.

Now  we compute $\Expec RAR^*$ for general operators $A$ using the resolvent identity and Gaussian integration by parts.
\begin{lemma}
\label{lem:T-eq-full}
The following identity holds for any $A\in \mcal{B}(\ell^2(\bbZ^d_L))$:
\begin{align*}
\Expec RAR^* &= \tilde{M}A\tilde{M}^* + \lambda^2 \tilde{M} \Expec \mcal{D}[RAR^*] \tilde{M}^*
+ \mathfrak{E}_T \\
\mathfrak{E}_T &= \tilde{M} A (\Expec R - \tilde{M})^*
- \lambda^2 \tilde{M} \Expec (\mcal{D}[R-\Expec R]) RAR^* + \lambda^2
\tilde{M} \Expec \mcal{D}[RAR^*] (R-\tilde{M})^*,
\end{align*}
\end{lemma}
\begin{proof}
To start, we use the resolvent identity
\[
R = \tilde{M} - \tilde{M} (\lambda V - \lambda^2\tilde{\theta}) R
\]
to expand
\begin{align*}
\Expec RAR^*
& = \Expec \tilde{M}AR - \Expec \tilde{M}(\lambda V- \lambda^2\tilde{\theta})RAR^*\\
& = \tilde{M}A\Expec R^*- \lambda \tilde{M}\Expec VRAR^* + \lambda^2 \tilde{\theta} M \Expec RAR^* .
\end{align*}
Then we apply Gaussian integration by parts to the expectation in the second term
\begin{align*}
\mathbb{E}VRAR^{\ast}
& = -\lambda \mathbb{E}\mathcal{D}[R]RAR^{\ast} - \lambda\mathbb{E}\mathcal{D}[RAR^{\ast}]R^{\ast}
\end{align*}
to get
$$\mathbb{E}RAR^{\ast} = \tilde{M}A\mathbb{E}R^{\ast} + \lambda^2 \tilde{M}\mathbb{E}\mathcal{D}(RAR^{\ast})R^{\ast} - \lambda^2\tilde{M}\mathbb{E}\mathcal{D}[R-\Expec R]RAR^{\ast}.$$
The result now follows from replacing the outer most $R^{\ast}$ in the first two terms with $\tilde{M}^*$.
\end{proof}

Taking $A = \sum_y a(y)\ket{y}\bra{y}$
to be a diagonal matrix, and defining $g(x) := \Expec (RAR^*)_{xx}$,
we see that the above gives the following equation for quantities of the form \eqref{eq:second-moment-sum}
\begin{equation}
\label{eq:elliptic-pde}
g = \tilde{K} \ast a + \lambda^2 \tilde{K} \ast g + \mathfrak{e}
\end{equation}
where $\mathfrak{e}(y) := (\mathfrak{E}_T)_{yy}$ and the convolution kernel $\tilde{K}$ is given by
\begin{equation*}
\tilde{K}(x) := |\Braket{0| \tilde{M}|x}|^2.
\end{equation*}
Letting $\tilde{\mcal{K}} u := \tilde{K}\ast u$ be the convolution operator by kernel $\tilde{K}$ and
iterating~\eqref{eq:elliptic-pde} yields the series expansion
\begin{equation}
\label{eq:b-eq}
\begin{aligned}
g &= (\Id + \lambda^2 \tilde{\mcal{K}} + \lambda^4 \tilde{\mcal{K}}^2 + \cdots ) (\mcal{K} f + \mathfrak{e}) \\
&= (\Id - \lambda^2\tilde{\mcal{K}})^{-1} (\mcal{K}f + \mathfrak{e}).
\end{aligned}
\end{equation}
The first equation above (in terms of the series expansion) gives $g$ in terms of the Green's function
of a random walk with step distribution $\tilde{K}$.  The second equation has the equivalent interpretation that $b$ is obtained as a solution to an elliptic equation with right hand side $\mcal{K}f$ -- this second interpretation is the one taken in Section~\ref{sec:RAR} below.
The main result of this section is an analysis of the error $(\Id-\lambda^2\mcal{K})^{-1}\mathfrak{e}$.
\begin{theorem}
\label{thm:diffusive-profile}
Let $d\geq 2$ and $\epsilon,\delta>0$. Then for any $A = \sum a(y)\ket{y}\bra{y}$ with $\|a\|_{\ell^\infty}\leq 1$
and $z=E+i\eta$ with $d(E,\Sigma_d)>\eps$ and $\eta\in(\lambda^{2+\kappa_d-\delta} ,\lambda^{2+\delta})$,
we have
\begin{equation}
\label{eq:RAR-flucs}
\sup_{x\in\Z^d_L}|(RAR^*)_{xx} - (\Id -\lambda^2 \tilde{\mcal{K}})^{-1} (\tilde{\mcal{K}}a)(x)| \ldom \eta^{-1} \Phi_d(z),
\end{equation}
where
\[
\kappa_d = \begin{cases}
\frac{2}{13}, &d\in\{2,4,5,6\} \\
\frac{2}{9}, & d=3 \\
\frac{2d-12}{3d-12}, & 7\leq d\leq 12 \\
\frac12, &d>12
\end{cases}
\]
and
\[
\Phi_d(z) =
\begin{cases}
\lambda^{1/2}(\lambda^2\eta^{-1})^{13/4}, &d\in\{2,4,5,6\} \\
\lambda^{3/4}(\lambda^2\eta^{-1})^{27/8}, &d= 3\\
\lambda(\lambda^2\eta^{-1})^2, &d\geq 7.
\end{cases}
\]
\end{theorem}

Note that for some  dimensions, $\kappa_d$ specifies a slightly worse range of $\eta$ compared to Theorem~\ref{thm:local-law}. Basically, it is chosen so that when $\eta\sim \lambda^{2+\kappa_d}$, $\Phi_d(z)\sim 1$.

The main error estimate is the following lemma.
\begin{lemma}
\label{lem:T-eq-error-estimate}
Let $\Phi_d$ and $z$ be as in Theorem \ref{thm:diffusive-profile} and $\mathfrak{E}_{\rm T}$ be as in Lemma \ref{lem:T-eq-full}, and $\tilde{K}$ as above.  Then for $\|A\|_{2\to 2} \leq 1$ we have
\begin{equation}
\label{eq:Emaxentry}
\|\mathfrak{E}_T\|_{1\to\infty}\leq C\lambda^{-2}\Phi_d(z).
\end{equation}
Moreover, the kernel $\tilde{K}$ satisfies
\begin{equation}
\label{eq:walk-mass}
\lambda^2 \sum_{x\in \Z^d} \tilde{K}(x)=1 - \lambda^{-2}\eta (\Impt M_{00})^{-1} + O(\lambda^{c-2}\eta).
\end{equation}
\end{lemma}

First we prove Theorem \ref{thm:diffusive-profile} assuming Lemma \ref{lem:T-eq-error-estimate} above.
\begin{proof}[Proof of Theorem \ref{thm:diffusive-profile} using Lemma~\ref{lem:T-eq-error-estimate}]
By~\eqref{eq:b-eq} we have
$$\Expec (RAR^*)_{xx} = (\Id -\lambda^2 \tilde{\mcal{K}})^{-1} (\mcal{K}a)(x)
+ (\Id-\lambda^2\tilde{\mathcal{K}})^{-1}\mathfrak{e}(x),$$
with $\mathfrak{e}(y) = (\mathfrak{E}_T)_{yy}$.  By \eqref{eq:walk-mass} and \eqref{eq:Emaxentry} we have
\begin{align*}
\|(\Id-\lambda^2\tilde{\mathcal{K}})^{-1}\mathfrak{e}\|_{\infty}
& \leq \|(\Id-\lambda^2\tilde{\mathcal{K}})^{-1}\|_{\infty\to\infty} \|\mathfrak{e}\|_{\infty}\\
& \lsim C\lambda^2\eta^{-1} \|\mathfrak{E}_T\|_{1\to \infty}\\
& \lsim \eta^{-1} \Phi_d(z),
\end{align*}
where we used that $\tilde{K}$ is a positive
kernel and \eqref{eq:walk-mass} to bound $\|(\Id-\lambda^2\mathcal{K})^{-1}\|_{\infty\to\infty}$.
By Lemma~\ref{lem:T-eq-full} it therefore follows that
\[
|\Expec (RAR^*)_{xx} - (\Id-\lambda^2\tilde{\mcal{K}})^{-1}(\tilde{\mcal{K}}a)(x)|
\lsim \eta^{-1}\Phi_d(z).
\]

To complete the proof of~\eqref{eq:RAR-flucs} we simply estimate the fluctuations of $(RAR^*)_{xx}$.
Lemma~\ref{lem:R-lsi} implies
\begin{align*}
|(RAR^*)_{xx} - \Expec (RAR^*)_{xx}| \ldom
\lambda \|R\|_{1\to 4} \|R\|_{1\to 2} \|R\|_{2\to 4} \|A\|_{2\to 2}.
\end{align*}

In dimension $d\geq 7$, Lemma~\ref{lem:R-flucs-highd} implies that
$\|R\|_{1\to 4} \ldom 1$, the Ward identity implies $\|R\|_{1\to 2}\ldom \eta^{-1/2}$,
and we estimate $\|R\|_{2\to 4}$ as follows using Proposition~\ref{prp:R_0Bounds} and Proposition \ref{lem:two-to-p-transfer}:
\begin{align*}
\|R\|_{2\to 4}
&\lsim \|R^{1/2}\|_{2\to 2} \|R^{1/2}\|_{2\to 4} \\
&\ldom \eta^{-1/2} (\lambda^2\eta^{-1})^{1/2}.
\end{align*}
Combining these gives (in $d\geq 7$):
\begin{align*}
|(RAR^*)_{xx} - \Expec (RAR^*)_{xx}| \ldom \eta^{-1} (\lambda (\lambda^2\eta^{-1})^{1/2}).
\end{align*}
If $d=3$ we estimate $\|R\|_{2\to 4}$ similarly for any $p>14/3$
\begin{align*}
\|R\|_{2\to 4}
&\lsim \|R^{1/2}\|_{2\to 2} \|R^{1/2}\|_{2\to 4} \\
&\lsim \|R^{1/2}\|_{2\to 2}^{9/8+\eps_p} \|R^{1/2}\|_{2\to p}^{7/8-\eps_p'} \\
&\ldom \eta^{-9/16} (\lambda^2\eta^{-1})^{7/16}.
\end{align*}
We also have $\|R\|_{1\to 2}\ldom \eta^{-1/2}$ by the Ward identity,
and estimating $\|R\|_{1\to 4}$ similarly gives, after some arithmetic,
\[
|(RAR^*)_{xx} - \Expec (RAR^*)_{xx}| \ldom
\eta^{-1} (\lambda^{3/4} (\lambda^2\eta^{-1})^{23/16}).
\]
Performing the same computation for $d=2,4,5,6$, we find that
\[
|(RAR^*)_{xx} - \Expec (RAR^*)_{xx}| \ldom
\eta^{-1} (\lambda^{1/2} (\lambda^2\eta^{-1})^{11/8}).
\]
In summary, we have obtained
\begin{align}
\label{eq:ComputedRAR-Flucs}
|(RAR^*)_{xx} - \Expec (RAR^*)_{xx}| \ldom
\eta^{-1}
\begin{cases}
\lambda^{1/2}(\lambda^2\eta^{-1})^{11/8}&d\in\{2,4,5,6\}\\
\lambda^{3/4}(\lambda^2\eta^{-1})^{23/16}&d=3\\
\lambda (\lambda^2\eta^{-1})^{1/2}&d\geq 7.
\end{cases}
\end{align}
In each case, the conclusion is that the fluctuations are much smaller than the error bound we obtain for the expectation,
so~\eqref{eq:RAR-flucs} holds.
\end{proof}

The rest of the section is dedicated to the proofs of~\eqref{eq:walk-mass} and~\eqref{eq:Emaxentry}.

\subsection{Proof of~\eqref{eq:walk-mass}}
Since $\tilde{M} = (\Delta - z - \lambda^2 \Expec R_{00})^{-1}$, the Ward identity implies
\begin{align*}
\lambda^2 \sum_{x\in\Z^d} \tilde{K}(x)
&= \lambda^2 \sum |(\Delta-z-\lambda^2\Expec R_{00})^{-1}_{0x}|^2 \\
&= \lambda^2
\frac{\Impt (\Delta - z -\lambda^2 \Expec R_{00})^{-1}_{00}}{\lambda^2 \Impt R_{00} + \eta}.
\end{align*}
One can check after some arithmetic that by applying Theorem~\ref{thm:local-law} there exists some $c>0$ such that $|\Expec R_{00} - M_{00}| \lsim \lambda^c(\lambda^{-2}\eta)$ for $\eta>\lambda^{2+\kappa_d-\delta}$.  Thus, using the smoothness of resolvent entries (Lemma~\ref{lem:dos-regularity}) and the definition of $M$, \eqref{eq:def-of-M},  we have
\begin{align*}
\lambda^2 \sum_{x\in\Z^d} \tilde{K}(x)
&= \frac{\lambda^2 \Impt M_{00} + O(\lambda^c\eta)}{\lambda^2 \Impt M_{00} + \eta + O(\lambda^c\eta)} \\
&= 1 - \lambda^{-2}\eta (\Impt M_{00})^{-1} + O(\lambda^{c-2}\eta),
\end{align*}
as claimed.

\subsection{Proof of the bound~\eqref{eq:Emaxentry}}

The bound on $\|\mathfrak{E}_T\|_{1\to\infty}$ is a straightforward application of  Theorem~\ref{thm:local-law},
the triangle inequality, and Cauchy Schwarz. We give the details below.
First we recall the expression for $\mathfrak{E}_T$ and apply
the triangle inequality:
\begin{equation}
\begin{split}
\|\mathfrak{E}_T\|_{1\to\infty}
&\leq \|\tilde{M} A (\Expec R - \tilde{M})^*\|_{1\to\infty}
+ \lambda^2 \|\tilde{M} \Expec \mcal{D}[R-\Expec R] RAR^*\|_{1\to\infty}
+ \lambda^2 \|\tilde{M} \Expec \mcal{D}[RAR^*] (R-\tilde{M})^*\|_{1\to\infty} \\
&=: \textrm{I} + \textrm{II} + \textrm{III}.
\end{split}
\end{equation}
To bound these error terms we will need an estimate for
$\|\Expec R-\tilde{M}\|_{1\to 2}$.
Using the Ward identity and Theorem~\ref{thm:local-law} for $\eta\in (\lambda^{2+\kappa_d},\lambda^{2})$ we have
\begin{equation}
\label{eq:ERM-one-two}
\|\Expec R - \tilde{M}\|_{1\to 2}
\lsim
\lambda^{-\delta}\times
\begin{cases}
\lambda^{-1/2} (\lambda^2\eta^{-1})^{9/4}, & d\in\{2,4,5,6\} \\
\lambda^{-1/4} (\lambda^2 \eta^{-1})^{19/8}, & d=3 \\
\lambda^2 \eta^{-1}, & d \geq 7.
\end{cases}
\end{equation}

Now we are ready to compute. First we estimate (using Corollary~\ref{cor:free-onetwo} for $\|\tilde{M}\|_{1\to 2}$ and Theorem~\ref{thm:local-law})
\[
\textrm{I} \leq \|\tilde{M}\|_{1\to 2} \|A\|_{2\to 2}
\|\Expec R -\tilde{M}\|_{1\to 2}
\lsim \lambda^{-2} \times \lambda^{-\eps}
\begin{cases}
\lambda^{1/2} (\lambda^2\eta^{-1})^{9/4}, & d\in\{2,4,5,6\} \\
\lambda^{3/4} (\lambda^2 \eta^{-1})^{19/8}, & d=3 \\
\lambda (\lambda^2 \eta^{-1}), & d \geq 7.
\end{cases}
\]
For the second term we use Theorem~\ref{thm:local-law} and the Ward identity to estimate $\|R\|_{1\to 2} \ldom \eta^{-1/2}$ and then write
\[
\textrm{II}
\lsim \lambda^2 \|\tilde{M}\|_{1\to 2}
\Expec \|R-\Expec R\|_{1\to\infty} \|R\|_{2\to 2} \|A\|_{2\to 2} \|R\|_{1\to 2}
\lsim \lambda^{-2} \times \lambda^{-\eps}
\begin{cases}
\lambda^{1/2} (\lambda^2\eta^{-1})^{13/4} , &d\in\{2,4,5,6\} \\
\lambda^{3/4} (\lambda^2\eta^{-1})^{27/8} , &d=3 \\
\lambda (\lambda^2\eta^{-1})^{3/2} , &d \geq 7.
\end{cases}
\]

The final term is estimated by splitting the expectation into two parts by adding and subtracting
$\Expec \mcal{D}[RAR^*]$:
\begin{equation}
\label{eq:three-split}
\begin{split}
\textrm{III}
&\leq \lambda^2 \|\tilde{M}\|_{2\to \infty}
\Expec
\|\mcal{D}[RAR^*] - \Expec \mcal{D}[RAR^*]\|_{2\to 2}
(\|R^*\|_{1\to 2}+\|\tilde{M}^*\|_{1\to 2}) \\
&\qquad + \lambda^2 \|\tilde{M}\|_{2\to\infty}
\|\Expec \mcal{D}[RAR^*]\|_{2\to 2}
\|\Expec (R^*-\tilde{M}^*)\|_{1\to2}.
\end{split}
\end{equation}
The first term above is bounded using~\eqref{eq:ComputedRAR-Flucs}
to estimate
\begin{equation}
\label{eq:four-bd}
\lambda^2 \|\tilde{M}\|_{2\to \infty}
\Expec \|\mcal{D}[RAR^*] - \Expec \mcal{D}[RAR^*]\|_{2\to 2}
(\|R^*\|_{1\to2}+\|\tilde{M}^*\|_{1\to 2})
\lsim
\lambda^{-2} \times \lambda^{-\eps}
\begin{cases}
\lambda^{1/2}(\lambda^2\eta^{-1})^{23/8}, &d\in\{2,4,5,6\} \\
\lambda^{3/4}(\lambda^2\eta^{-1})^{47/16}, &d=3 \\
\lambda (\lambda^2\eta^{-1})^2, &d\geq 7.
\end{cases}
\end{equation}
For the latter term in~\eqref{eq:three-split}, we estimate $\|\mcal{D}[RAR^*]\|_{2\to 2}
\leq \|A\|_{2\to 2} \|R\|_{1\to 2}^2 \ldom \eta^{-1}$ and use~\eqref{eq:ERM-one-two} to bound
\begin{equation}
\label{eq:six-bd}
\lambda^2 \|\tilde{M}\|_{1\to 2} \|\Expec\mcal{D}[RAR^*]\|_{2\to2} \|\Expec (R^*-\tilde{M}^*)\|_{1\to 2}
\lsim \lambda^{-2} \times \lambda^{-\eps}
\begin{cases}
\lambda^{1/2} (\lambda^2\eta^{-1})^{13/4} , &d\in\{2,4,5,6\} \\
\lambda^{3/4} (\lambda^2\eta^{-1})^{27/8} , &d=3 \\
\lambda (\lambda^2\eta^{-1})^{3/2} , &d\geq 7.
\end{cases}
\end{equation}
Adding~\eqref{eq:four-bd} and~\eqref{eq:six-bd} we conclude
\[
\textrm{III}
\lsim \lambda^{-2} \times \lambda^{-\eps}
\begin{cases}
\lambda^{1/2} (\lambda^2\eta^{-1})^{13/4} , &d\in\{2,4,5,6\} \\
\lambda^{3/4} (\lambda^2\eta^{-1})^{27/8} , &d=3 \\
\lambda (\lambda^2\eta^{-1})^{2} , &d\geq 7.
\end{cases}
\]
Collecting the worst bound on ${\rm I},\rm{II},$ and $\rm{III}$ in each dimension, we may conclude.

\section{Analysis of the diffusion profile}
\label{sec:RAR}
In the last section, we saw that $ (RAR^*)_{xx}$ is approximately given in terms of the operator
\begin{align}
\label{eq:mcalGdef}
	\mathcal{G}(z):&=(\Id-\lambda^{2}\tilde{\mathcal{K}}(z))^{-1}\tilde{\mathcal{K}}(z),
\end{align}
in the sense that, for $a\in \ell^\infty(\bbZ^d_L)$ and $A = \sum a(x)\ket{x}\bra{x}$,
\[
(RAR^*)_{xx} = (\mcal{G}(z) a)(x) + O(\lambda^c \eta^{-1} \|a\|_{\ell^\infty}).
\]
In this section we compare $\mcal{G}$
to the inverse of elliptic operator on $\Real^d$.  To do this it is convenient to change the domain from $\bbZ^d_L$ to $\bbZ^d$.

This involves a change of notation, so we write $\tilde{K}_{\bbZ^d}(z)$ for the corresponding convolution operator on $\bbZ^d$,
\[
\tilde{K}_{\bbZ^d}(z;x) = |(\Delta_{\bbZ^d} - (z+\lambda^2\tilde{\theta}))^{-1}_{0x}|^2.
\]
Because of Lemma~\ref{prp:Zd-to-ZdL}, we have
$\tilde{K}_{\bbZ^d}(z;x) - \tilde{K}(z;x) = O(\exp(-c\lambda^{-1})$ for $|x|\ll \lambda^{-10}$.  Similarly, we write $\mcal{G}_{\bbZ^d}$ for the operator defined as in~\eqref{eq:mcalGdef} with $\tilde{\mcal{K}}_{\bbZ^d}$ replacing $\tilde{\mcal{K}}$.

Since $\tilde{K}$ is approximately $\tilde{K}_{\bbZ^d}$ near the origin, we therefore also have for $\supp a \subset \{|x|\ll L\}$ and $A_{\bbZ^d} = \sum_{\bbZ^d} a(x)\ket{x}\bra{x}$,
\[
(RAR^*)_{xx} =
(\mcal{G}_{\bbZ^d}(z) a)(x) + O(\lambda^c\eta^{-1}).
\]

We will show that, on functions smooth to scale $\lambda^{-2}$, $\mcal{G}_{\bbZ^d}$ approximately acts as the inverse of an elliptic operator.
To define this operator, we let $m(z)$ be the ``mass''
\begin{equation}
\label{eq:massdef}
m(z) := \lambda^2\eta^{-1} (1 - \lambda^2 \sum_{x\in \bbZ^{d}} \tilde{K}(z;x)).
\end{equation}
and $\vartheta(z)$ be the diffusion constant
\begin{equation}
\label{eq:diffdef}
\vartheta(z) := \frac{\lambda^6}{2}  \sum_x x_1^2 \tilde{K}(z;x).
\end{equation}
We then define the elliptic operator $\mcal{L}$ on $\Real^d$ as
\begin{equation}
\label{eq:L-def}
\mcal{L} := -\lambda^{-4}\vartheta  \Delta_{\Real^d} +  \lambda^{-2}\eta m,
\end{equation}
where we suppress the dependence on $z$.   The factors of $\lambda$ and $\eta$ above are chosen so that $m(z)$ and $\vartheta(z)$ are of unit order $1$, as computed below in~\eqref{eq:massdiff-limit}.

More concretely, we will compare $S\mcal{L}^{-1} $ to  $ \mathcal{G}S$, where $S:C^0(\Real^d)\to \ell^\infty(\bbZ^d)$ is the ``sampling'' operator
that simply evaluates a function on the lattice points.   We also give approximations to $m(z)$ and $\vartheta(z)$ in terms
primitive functions involving only the dispersion relation of the Laplacian. The first function is the density of states $\rho$ defined in~\eqref{eq:rho-def}.
The other function we need is $\nu:[-2d,2d]\to\Real_+$ given by
\begin{equation}
\label{eq:nu-def}
\nu(E) := (2\pi)^{-d}  \int_{\omega(p)=E} |\nabla \omega(p)| \diff\mcal{H}^{d-1}(p).
\end{equation}

We can now state our main result of the section.
\begin{proposition}
\label{prp:elliptic-limit}
Let $\mathcal{G}_{\bbZ^d}(z):\ell^\infty(\bbZ^{d})\to \ell^{\infty}(\bbZ^{d})$ be the convolution operator~\eqref{eq:mcalGdef},
let $\mathcal{L}:C^\infty(\Real^d)\to C^\infty(\Real^d)$  be the
elliptic operator defined as in~\eqref{eq:L-def}, and let $S:C^0(\Real^d)\to \ell^\infty(\bbZ^d)$ be the sampling operator.
Let $z=E+i\eta$ with $\lambda^{2+\kappa_d-\delta} < \eta < \lambda^{2+\delta}$.
Then for $f\in C_c^\infty(\Real^d)$ with $\supp f\subset\{|x|\leq \frac{1}{100}L\}$ we have
\[
\|\lambda^{-2}S \mathcal{L}^{-1} f - \mathcal{G}_{\bbZ^d} S f\|_{\ell^\infty(\bbZ^{d})} \lsim \lambda^{-2} ( \eta^{-1}\|\nabla f\|_{C^0(\R^{d})} + \|\Ft{f}\|_{L^1(\R^{d})}).
\]
Moreover, we have the approximations
\begin{equation}
\begin{split}
\label{eq:massdiff-limit}
m(z) &= \rho(E)^{-1} +  O(\lambda^{c})\\
\vartheta(z) &= \frac{\pi}{4d}\rho(E)^{-3}\nu(E) + O(\lambda^{c}).
\end{split}
\end{equation}
\end{proposition}

\subsection{Computation of the moments of $K$}
Recall that the kernel $\tilde{K}_{\bbZ^d}(z;x)$ is given by
\[
\tilde{K}_{\bbZ^d}(x) = |\Braket{0| (\Delta - (z+\lambda^2\tilde{\theta}(z))|x}|^2
\]
with\footnote{Note $\tilde{\theta}$ is technically defined in terms of the resolvent on $\bbZ^d_L$, although this this does not matter for us since we have $|\tilde{\theta}-\tilde{\theta}|\lsim \exp(-c\lambda^{-1})$  anyway.} $\tilde{\theta}(z)=\Expec R(z)_{00}$.
We first compute its symbol at low frequencies.
\begin{lemma}
\label{lem:KSymbol}
For $z=E+i\eta$ with $d(E,\Sigma_d) > \eps$ and $\eta>\lambda^{2+\kappa_d-\delta}$, we have that the kernel $\tilde{K}(z;x)$ satisfies
\begin{align*}
	\lambda^{2}\widehat{\tilde{K}}(\xi)=(1-\lambda^{-2}\eta m)-\vartheta\lambda^{-4}|\xi|^{2}+ \lambda^{-8} |\xi|^4 r(\xi),
\end{align*}
where $|r(\xi)| \lsim 1$ for all $\xi$.
\end{lemma}
\begin{proof}
To begin, observe that since $\Delta_{\mathbb{Z}^d}$ is invariant under reflections,
$\tilde{K}_{\mathbb{Z}^d}(x)=\tilde{K}_{\mathbb{Z}^d}(x')$ if $x'$ is the reflection of $x$ across a hyperplane of
the form $\{(x_1,...,x_d)\in \mathbb{R}^d\ | x_i = 0\}$.
Therefore, odd moments of $\tilde{K}_{\mathbb{Z}^d}(x)$ vanish, so in particular,
\begin{align*}
	\sum_{x\in \bbZ^{d}}x_i\tilde{K}_{\mathbb{Z}^d}(x)=\sum_{x\in \bbZ^{d}}x_ix_j\tilde{K}_{\mathbb{Z}^d}(x)=\sum_{x\in \bbZ^{d}}x_ix_jx_k\tilde{K}_{\mathbb{Z}^d}(x)=0,
\end{align*}
for any indices $i,j,k$ with $i\neq j$. Taylor expanding $e^{ix\xi}$ to third order and recalling the definitions of $m$ and $\vartheta$, we find that
\begin{align*}
\lambda^2\hat{\tilde{K}}_{\mathbb{Z}^d}(x)=(1-\lambda^{-2}\eta m)-\vartheta\lambda^{-4}|\xi|^{2} +\lambda^{2}\sum_{x\in\bbZ^{d}}r_1(x\xi)\tilde{K}_{\mathbb{Z}^d}(x),
\end{align*}
with $|r_{1}(y)|\leq C|y|^{4}$ for all $y\in \Real$.

It remains to show that
\[
\lambda^2 \sum_{x\in\bbZ^d} |x|^4 \tilde{K}_{\mathbb{Z}^d}(x)
\lsim \lambda^{-8}.
\]
By Plancherel's theorem, this is equivalent to the bound
\begin{equation}
\label{eq:four-planch-bd}
\sum_{j,k=1}^d\int_{\Torus^d} \Big|\partial_j\partial_k \frac{1}{\omega(\xi)- (z+ \lambda^2\theta)}\Big|^2 \diff \xi
\lsim \lambda^{-10}.
\end{equation}
For convenience we set $z' = z +\lambda^2\tilde{\theta}(z) $, and we note $\Impt z' \gtrsim \lambda^2$.
Then we compute
\begin{align*}
\sum_{j,k=1}^d\int_{\Torus^d} \Big|\partial_j\partial_k \frac{1}{\omega(\xi)-z'}\Big|^2 \diff \xi
&\lsim \sum_{j,k} \int_{\Torus^d} \Big( \frac{|\partial_j \omega|^2 |\partial_k\omega|^2}{|\omega(\xi)-z'|^6}
+ \frac{|\partial_{jk}\omega|^2}{|\omega(\xi)-z'|^4}\Big)\diff \xi \\
&\lsim \int_{\Torus^d} \frac{1}{|\omega(\xi)-z'|^6} \diff \xi.
\end{align*}
To get to the second line we used that $\omega\in C^2(\Torus^d)$ and that $|\omega(\xi)-z'|\lsim 1$ to bound $|\omega(\xi)-z'|^{-4} \lsim |\omega(\xi)-z'|^{-6}$.  Then we use the coarea formula and change variables to write
\begin{align*}
\int_{\Torus^d} \frac{1}{|\omega(\xi)-z'|^6}
\diff \xi
= \int \frac{1}{|E'-i\Impt z'|^6}\rho(E'-\Rept z')\diff E'.
\end{align*}
The conclusion~\eqref{eq:four-planch-bd} now follows from observing that
\[
\int_{-\infty}^\infty
\frac{1}{|t-i\Impt z'|^6} \diff t
\lsim (\Impt z')^{-5} \lsim \lambda^{-10},
\]
and that $\rho(E'-\Rept z')$ is bounded near the origin (for $d(E,\Sigma_d)>\eps$).

\end{proof}
Now we prove approximations the for $m$ and $\vartheta$ recorded in \eqref{eq:massdiff-limit}.
\begin{proposition}
\label{pr:massThetaApprox}
For any $z=E+i\eta$ with $d(E,\Sigma_d)>\eps$ and $\eta>\lambda^{2+\kappa_d-\delta}$, we have that
\begin{align}
m(z)&=\rho(E)^{-1}+O(\lambda^{c})\label{eq:m-approx}\\
\vartheta(z)&=\frac{\pi}{4d}\rho(E)^{-3}\nu(E)+O(\lambda^{c})\label{eq:theta-approx}.
\end{align}
\end{proposition}
\begin{proof}
By \eqref{eq:walk-mass}, we have that
\begin{align*}
	m(z)=(\Impt \theta(z))^{-1}+O(\lambda^{c}).
\end{align*}
By Lemma~\ref{pr:thetaAPriori} we have the estimate for $\theta(z)$
\begin{align}\label{eq:simpleThetaEst}
|\Impt(\theta(z))-  \rho(E)| \lsim \lambda^2.
\end{align}
Since $\Impt \theta(z)$ and  $\rho(E)$ are bounded below, \eqref{eq:m-approx} now follows.

To prove \eqref{eq:theta-approx}, we compute using the Plancherel theorem and the coarea formula that for $u\in \bbH$ with $d(\Rept u,\Sigma_d)>\eps$
\begin{align*}
(\Impt u)^{3}\sum_{x\in\bbZ^d} |x|^2 |\braket{0|(\Delta_{\bbZ^{d}}-u)^{-1}|x}|^{2}
&= (2\pi)^{-d} \int_{\Torus^d} \frac{(\Impt u)^{3}|\nabla \omega(\xi)|^2}{|\omega(\xi) - u|^4} \diff \xi \\
&= \int_{-\infty}^{\infty} \frac{(\Impt u)^3}{|E'-u|^4}  \nu(E') \diff E'.
\end{align*}
By the continuity of $\nu(E)$ and the integral
\[
\int_{-\infty}^\infty \frac{(\Impt u)^3}{|E'-u|^4} \diff E' = \frac{\pi}{2},
\]
it follows that
\begin{align*}
\sum_{x\in\bbZ^d} |x|^2 |\braket{0|(\Delta_{\bbZ^{d}}-u)^{-1}|x}|^{2}=\frac{\pi}{2}\nu(\Rept u)(\Impt u)^{-3}+O((\Impt u)^{-2}).
\end{align*}
Recalling the definition of $\vartheta$ and $\tilde{K}$, we have shown that
\begin{align*}
	\vartheta(z)=\frac{\pi}{4d}\nu(E+\lambda^{2}\Rept \tilde{\theta}(z))/(\Impt \tilde{\theta}(z))^{3}+O(\lambda^{6}(\eta+\lambda^{2}\Impt \tilde{\theta}(z))^{-2}).
\end{align*}
The approximation \eqref{eq:theta-approx} now follows from \eqref{eq:simpleThetaEst} and Theorem~\ref{thm:local-law}.
\end{proof}

\subsection{Proof of Proposition~\ref{prp:elliptic-limit}}

\begin{proof}
It is easy to see that
\begin{align*}
\mathcal{G}_{\mathbb{Z}^d}= \lambda^{-2}\left( (\Id-\lambda^{2}\tilde{\mcal{K}}_{\mathbb{Z}^d})^{-1}-\Id \right)
\end{align*}
so it suffices to prove that
\begin{equation}
\label{eq:actual-bd}
\|(\Id-\lambda^{2}\tilde{\mathcal{K}}_{\mathbb{Z}^d})^{-1} S f - S\mathcal{L}^{-1} f\|_{\ell^\infty} \lsim \eta^{-1} \|\nabla f\|_{C^0} + \|\Ft{f}\|_{L^1}.
\end{equation}
Before continuing, we stop to make a few observations.  The first is that we have the simple $\ell^\infty\to\ell^\infty$ bound
\begin{equation}
\label{eq:latticeL-apriori}
\|(\Id-\lambda^{2}\tilde{\mcal{K}}_{\mathbb{Z}^d})^{-1}\|_{\ell^\infty\to\ell^\infty} \lsim \lambda^2\eta^{-1},
\end{equation}
which follows from the Neumann series expansion
\[
(\Id-\lambda^{2}\tilde{\mathcal{K}}_{\mathbb{Z}^d})^{-1} = \Id + \lambda^2\mcal{\tilde{K}}_{\mathbb{Z}^d} + \lambda^4 \mcal{\tilde{K}}_{\mathbb{Z}^d}^2 + \cdots
\]
and \eqref{eq:walk-mass}.
A similar argument using the representation $\mcal{L}^{-1} = \int_0^\infty e^{-t\mcal{L}} \diff t$ shows that
\begin{equation}
\label{eq:Lapriori}
\|\mcal{L}^{-1}\|_{L^\infty\to L^\infty} \lsim \lambda^2\eta^{-1}.
\end{equation}
Next we simplify the problem by taking a low frequency cutoff in Fourier space.
Let $\chi \in \mcal{S}(\Real^d)$ be a smooth function with compact Fourier support $\supp\Ft{\chi}\subset B_1$ and $\int \chi = 1$.  For any $\ell>0$ define $\chi_\ell(x) := \ell^{-d} \chi(x/\ell)$, which has Fourier support in $B_{1/\ell}$.
Then~\eqref{eq:latticeL-apriori} and~\eqref{eq:Lapriori} imply
\begin{equation}
\label{eq:smoothingerr1}
\|(\Id-\lambda^{2}\tilde{\mcal{K}}_{\mathbb{Z}^d})^{-1} S f - (\Id-\lambda^{2}\tilde{\mcal{K}}_{\mathbb{Z}^d})^{-1} S(\chi_\ell \ast f)\|_{\ell^\infty} \lsim \lambda^2\eta^{-1} \ell \|\nabla f\|_{C^0}
\end{equation}
and
\begin{equation}
\label{eq:smoothingerr2}
\|S \mcal{L}^{-1} f - S\mcal{L}^{-1} (\chi_\ell\ast f)\|_{\ell^\infty} \lsim \lambda^2 \eta^{-1} \ell \|\nabla f\|_{C^0}.
\end{equation}
We will take $\ell = C\lambda^{-2}$, so it suffices to show
\begin{equation}
\label{eq:interesting-bd}
\|S \mcal{L}^{-1} (\chi_\ell \ast f) - (\Id-\lambda^{2}\tilde{\mathcal{K}}_{\mathbb{Z}^d})^{-1} S(\chi_\ell \ast f)\|_{\ell^\infty} \lsim \|\Ft{f}\|_{L^1}.
\end{equation}
The Fourier transform $\mcal{F}$ of the first term is given by
\[
\mcal{F}S\mcal{L}^{-1}(f\ast \chi_\ell)(\xi) = \frac{1}{\lambda^{-2}\eta m+\vartheta \lambda^{-4}|\xi|^2} \Ft{\chi}(\ell\xi) \Ft{f}(\xi)
\]
so by Lemma~\ref{lem:KSymbol}, we find that
\begin{equation}
\begin{split}
\label{eq:ft-bd}
&\|S\mcal{L}^{-1}(\chi_\ell\ast f) - (\Id-\lambda^{2}\tilde{\mcal{K}}_{\mathbb{Z}^d})^{-1}S(\chi_\ell\ast f_{\bbZ^d})\|_{\ell^\infty}\leq \\
&\quad\int_{\Real^{d}}
\Big|
\frac{1}{\lambda^{-2}\eta m + \vartheta \lambda^{-4} |\xi|^2 + r(\xi)\lambda^{-8}|\xi|^4}
-
\frac{1}{\lambda^{-2}\eta m+\vartheta \lambda^{-4}|\xi|^2}\Big| |\Ft{\chi}(\ell\xi)||\Ft{f}(\xi)| \diff \xi.
\end{split}
\end{equation}
Now since $\ell=C\lambda^{-2}$, the integral is supported on $|\xi|\leq 1/C$.
For such $\xi$, (taking $C$ large enough), we have the bound
\[
\Big|
\frac{1}{\lambda^{-2}\eta m + \vartheta \lambda^{-4} |\xi|^2 + r(\xi)\lambda^{-8}|\xi|^4}
-
\frac{1}{\lambda^{-2} \eta m+\vartheta \lambda^{-4}|\xi|^2}\Big| \lsim 1.
\]
Applying this bound into~\eqref{eq:ft-bd} and integrating proves~\eqref{eq:interesting-bd}.
Finally, note that the approximations~\eqref{eq:massdiff-limit} have been shown in Proposition~\ref{pr:massThetaApprox}, so the proof is complete.
\end{proof}

\section{Proofs of the main results}
\label{sec:proofs}
We are now ready to prove the main theorems stated in the introduction.  These all follow in some way
from the following result, which is an easy corollary of Theorem~\ref{thm:diffusive-profile} and Proposition~\ref{prp:elliptic-limit}. Note that while Theorem~\ref{thm:diffusive-profile} is formulated for $R=R_{\bbZ^d_L}$ we may use Proposition~\ref{prp:Zd-to-ZdL} to pass to $R_{\bbZ^d}$.
\begin{corollary}
\label{cor:main-cor}
Let $z=E+i\eta$ with $d(E,\Sigma_d)>\eps$ and $\eta > \lambda^{2+\kappa_d-\delta}$,
and let $f\in C_c^\infty(\Real^d)$ with $\supp f\subset \{|x|\leq \frac{1}{100}L\}$.  Define $u$ to be the solution of the equation
\begin{equation}
\label{eq:uf-eq}
(-\lambda^{-4}\vartheta(z) \Delta_{\Real^d} + \lambda^{-2}\eta m(z))u = \lambda^{-2}f,
\end{equation}
where $\vartheta$ and $m$ are defined in~\eqref{eq:massdef} and~\eqref{eq:diffdef}
and satisfy the approximations~\eqref{eq:massdiff-limit}.  Then
\begin{equation}
\label{eq:combined-compare}
|\sum_{x\in\Z^d} f(x) |(R_{\bbZ^d}(z))_{0x}|^2  - u(0)| \ldom
\lambda^{-2} ( \eta^{-1}\|\nabla  f\|_{C^0} + \|\Ft{f}\|_{L^1}) + \eta^{-1} \Phi_d(z),
\end{equation}
where $\Phi_d(z)$ is defined as in Theorem~\ref{thm:diffusive-profile}.
\end{corollary}

It is useful to scale this as $f=f_0(x/\ell)$ for
$\ell = \alpha \lambda^{-1}\eta^{-1/2}$ and $\alpha$ some free parameter.  When $\alpha\sim 1$,
$\ell$ is at the diffusive scale (the scale of $R(E+i\eta)$),
and when $\alpha\ll 1$, $f$ probes more local averages.
In this case, defining
$\tilde{u}(x) = \lambda^{-2}\ell^{-2} u(\ell x)$, the equation~\eqref{eq:uf-eq} becomes
\begin{equation}
\label{eq:rescaled-elliptic}
-\vartheta(z) \Delta_{\Real^d} \tilde{u} + \alpha^2 m(z)\tilde{u} = f_0.
\end{equation}
By~\eqref{eq:massdiff-limit} the factor $m(z)$ is order $1$.  Therefore in the regime
$\alpha \ll 1$, the mass term is negligible.

We will also need the following simple consequence of Corollary \ref{cor:main-cor}.
\begin{lemma}
\label{lem:radial-cutoff}
Let $\chi\in C_c^\infty(\Real^d)$ be a positive bump function satisfying $\chi \geq 1$ on $B_1$
and $\supp \chi \in B_2$.  Then for each $E\in \R$ with $d(E,\Sigma_d)>\eps$ and $0<r\ll L/100$, there is a continuous function $\phi:\Real_+\to\Real_+$ depending on $\chi$ and $E$ such that, for $z=E+i\eta$ with
$\eta\in (\lambda^{2+\kappa_d-\delta},\lambda^2)$,
\[
|\eta \sum_{x\in\bbZ^d} \chi(x / (r\lambda^{-1}\eta^{-1/2})) |(R_{\bbZ^d}(z))_{0x}|^2 - \phi(r)|
\ldom \Phi_d(z) + \lambda^{-1}\eta^{1/2},
\]
with implicit constants depending on $\chi$.
The function $\phi$ satisfies for small $r<1$
\begin{equation}
\label{eq:phi-smallr}
|\phi(r)| \lsim r^2
(1 + \log(r^{-1})\One_{d=2})
\end{equation}
whereas for large $r$ we have the estimate
\begin{equation}
\label{eq:phi-rho}
|\phi(r) - \rho(E)| \lsim e^{-cr}.
\end{equation}
\end{lemma}
 A consequence of the above bound is there exist constants $c$ and $C$  such that at most
 $1\
$99\

\begin{proof}
Let $\phi(r)$ be the function
\[
\phi(r) = \int_{\Real^d} \chi(x/r) G(x) \diff x,
\]
where $G$ is the Green's function for the operator
$(-\frac{\pi}{4d}\rho^{-3}(E)\nu(E) \Delta_{\Real^d} + \rho^{-1}(E))$.  The bound~\eqref{eq:phi-smallr} follows from a comparison to the Green's function for the Laplacian. Then the bound~\eqref{eq:phi-rho} follows from the fact that $G$ integrates to exactly
$\rho(E)$, and decays exponentially at length scale $\sqrt{\nu(E)}/\rho(E)$.

Now applying
Corollary~\ref{cor:main-cor} with $f(x) = \chi(x/(r \lambda^{-1}\eta^{-1/2}))$, we obtain
\begin{align*}
|\eta \sum_{x\in \bbZ^d} |\chi(x/(r\lambda^{-1}\eta^{-1/2}))
|R_{0x}(E+i\eta)|^2 - \phi(r)|
&\ldom (\lambda^{-1} \eta^{1/2}) \|\nabla \chi\|_{C^0} +
\lambda^{-2}\eta \|\chi\|_{C^0}
+ \Phi_d(z)\\
&\lsim \lambda^{-1} \eta^{1/2}+\Phi_d(z),
\end{align*}
as desired.
\end{proof}

\subsection{Proof of Theorem~\ref{thm:main-prop}}

The key identity behind the proof of Theorem~\ref{thm:main-prop} is the following consequence
of the Plancherel theorem, valid for any $\psi\in \ell^2(\Z^d)$ and $\eta>0$
\begin{equation}
\label{eq:parseval-again}
\int_0^\infty e^{-2\eta t} |e^{-itH}\psi(x)|^2 \diff t
= (2\pi)^{-1}\int_{-\infty}^\infty |R(E+i\eta)\psi (x)|^2 \diff E,
\end{equation}
see, e.g., \cite{reed1978iv}[Sec. XIII.7].
In particular, for any $f\geq 0$ we have
\[
\frac1T\int_0^T \sum_{x\in\Z^d} f(x)|(e^{-itH})_{0,x}|^2
\lsim \frac1T \int_{-\infty}^\infty \sum_{x\in\Z^d} f(x) |R(E+i T^{-1})_{0x}|^2\diff E.
\]
Both~\eqref{eq:averaged-prop-in-ball} and~\eqref{eq:lower-bound-prop} follow from bounding the
right hand side above for appropriate choices of $f$.   In fact, we will show that for
an appropriate choice of constants $c$ and $C$ (depending on the energy $E$ and the error
threshold $\delta$),
\begin{align}
\label{eq:not-in-little-ball}
\eta \int_{-\infty}^\infty\sum_{|x|\leq c\lambda^{-1}\eta^{-1/2}} |R(E+i\eta)_{0x}|^2 \diff E
\leq \frac12 \\
\label{eq:in-big-ball}
\eta \int_{-\infty}^\infty
\sum_{|x|\geq C\lambda^{-1}T^{1/2}} |(e^{-itH})_{0x}|^2 \leq \delta.
\end{align}
Then~\eqref{eq:averaged-prop-in-ball} follows from~\eqref{eq:not-in-little-ball} and the identity
\[
\frac1T\int_0^T \sum_{x\in\bbZ^d} |(e^{-itH})_{0x} |^2 = 1,
\]
whereas~\eqref{eq:in-big-ball} directly implies~\eqref{eq:lower-bound-prop}.

By taking $r\ll 1$ in Lemma~\ref{lem:radial-cutoff}, we find that for any $E$ with $d(E,\Sigma_d)>\eps$,
\begin{align*}
\eta \sum_{|x|\leq c\lambda^{-1}\eta^{-1/2}} |R(E+i\eta)_{0x}|^2 \diff E
\leq \frac12
\end{align*}
with probability at least $1-C\lambda^{10000}$. Similarly, for such $E$, because $\eta\sum_{x\in\Z^d}|(R(z))_{0x}|^2=\rho(E)+O(\lambda^c)$ with high probability by the Ward identity and the local law Theorem~\ref{thm:local-law}
, choosing $r\gg 1$ yields
\begin{align*}
\eta \int_{-\infty}^\infty
\sum_{|x|\geq C\lambda^{-1}T^{1/2}} |(R(E+i\eta))_{0x}|^2 \leq \delta,
\end{align*}
with high probability.
Moreover, the suprema over $E$ with $d(E,\Sigma_d)>\eps$ in both expressions is similarly bounded with probability $1-C\lambda^{1000}$ because the resolvent entries are $\eta^{-2}$-Lipschitz.

With this in mind, the proof of Theorem~\ref{thm:main-prop} now reduces to the following claim:
\begin{lemma}
\label{lem:integral-to-max}
For any positive $f\in\ell^\infty(\bbZ^d)$ and $\eps>0$ we have
\[
\eta \int_{-\infty}^\infty \sum_x f(x) |R(E+i\eta)_{0x}|^2\diff E
\lsim \eps^{1/5} \|f\|_{\ell^\infty} + \Big(\sup_{\substack{E\in [-2d,2d]\\ d(E,\Sigma_d)>\eps}}
\eta \sum_x f(x) |R(E+i\eta)_{0x}|^2 \Big).
\]
\end{lemma}
\begin{proof}[Proof of Lemma~\ref{lem:integral-to-max}]
Let $\psi_0$ be the Kronecker delta at the origin.  We write
\[
\psi_0 = \Pi^{\rm good} \psi_0 + \Pi^{\rm small}\psi_0
=: \psi^{\rm good} + \psi^{\rm small}.
\]
where $\Pi^{\rm good}$ is a smooth spectral projection to the set
\[
\{E \in [-2d,2d] \mid d(E,\Sigma_d) > \eps^{1/4}\}.
\]
The spectral projection $\Pi^{\rm small}$ is smooth to scale $\eps^{1/4}$, so by Proposition~\ref{prp:zero-bd} we have
\[
\|\Pi^{\rm small}\psi_0 - \Pi^{\rm small}(\Delta)\psi_0\| \lsim (\log \lambda^{-1})^C
\lambda \eps^{-1/8} \ll \eps^{1/5},
\]
where $\Pi^{\rm small}(\Delta_{\Z^d})$ is the corresponding
projection of the free Laplacian.  Then by the density of states estimate in Lemma~\ref{lem:free-dos}, we can conclude that
\[
\|\Pi^{\rm small}\psi_0\| \lsim \eps^{1/4}|\log(\eps^{1/4})| \ll \eps^{1/5}.
\]
It therefore suffices to show that
\[
\eta \sum_x \int_{-\infty}^\infty f(x) |R(E+i\eta)\psi^{\rm good}(x)|^2 \diff E
\lsim
\eps^{1/5}\|f\|_{\ell^\infty}
+ \Big(\sup_{\substack{E\in [-2d,2d]\\ d(E,\Sigma_d)>\eps}}
\eta \sum_x f(x) |R(E+i\eta)_{0x}|^2 \Big).
\]
Now let $A := \Real \setminus \{E \in [-2d,2d] \mid d(E,\Sigma_d) > \eps\}$.  First we estimate
\[
\eta \int_{\Real \setminus A} \|R(E+i\eta)\psi^{\rm good}\|^2 \diff E \lsim \eps^{1/5}.
\]
Indeed,
\begin{align*}
\eta \int_{\Real\setminus A} \|R(E+i\eta)\psi^{\rm good}\|^2 \diff E
\leq &\eta \int_{\Real\setminus[-2d,2d]} \|R(E+i\eta)\psi^{\rm good}\|^2 \diff E \\
&+ \eta \sum_{E_j\in\Sigma_d}\int_{E_j-\eps}^{E_j+\eps} \|R(E+i\eta)\psi^{\rm good}\|^2 \diff E.
\end{align*}
For the infinite integral we use the following bound for $E\in \Real\setminus [-2d,2d]$:
\[
\|R(E+i\eta)\psi^{\rm good}\|^2 \lsim d(E,[-2d+\eps^{1/4},2d-\eps^{1/4}])^{-2},
\]
so that
\[
\eta \int_{\Real\setminus[-2d,2d]} \|R(E+i\eta)\psi^{\rm good}\|^2 \diff E \lsim \eta \eps^{-1/4}
\lsim \eps^{1/5} ,
\]
the latter bound holding for $\eta \lsim \eps^{1/2}$ for example.
For the finite integrals on $E\in [E_j-\eps,E_j+\eps]$ with $E_j\in\Sigma_d$ we use
\[
\|R(E+i\eta) \psi^{\rm good}\| \lsim \eps^{-1/4},
\]
so that
\[
\eta \int_{E_j-\eps}^{E_j+\eps} \|R(E+i\eta)\psi^{\rm good}\|^2 \diff E \lsim \eta \eps^{1/2}.
\]
What remains is to prove the bound
\[
\int_A \eta\sum_x f(x) |R(E+i\eta)\psi^{\rm good}(x)|^2 \diff E  \lsim
\eps^{1/5}\|f\|_{\ell^\infty}
+ \Big(\sup_{\substack{E\in [-2d,2d]\\ d(E,\Sigma_d)>\eps}}
\eta \sum_x f(x) |R(E+i\eta)_{0x}|^2 \Big).
\]
Now again using $\psi_0 = \psi^{\rm good} + \psi^{\rm bad}$ with $\|\psi^{\rm bad}\|\lsim \eps^{1/5}$
we have
\begin{align*}
\label{eq:final-calc}
\int_A \eta\sum_x f(x) |R(E+i\eta)\psi^{\rm good}(x)|^2 \diff E
&\lsim \eps^{1/5}\|f\|_{\ell^\infty} +
\int_A \eta \sum_x f(x) |R(E+i\eta)_{0x}|^2 \diff E. \\
&\lsim \eps^{1/5}\|f\|_{\ell^\infty} +
4d \sup_{E\in A} \eta \sum_x f(x)|R(E+i\eta)_{0x}|^2,
\end{align*}
as desired.

\end{proof}

\subsection{Proof of Theorem~\ref{thm:scaling}}
\label{sec:scaling-pf}

Fix $f_0\in C_c^\infty$, $\kappa<\kappa_d$, and $z\in[-2d,2d]\setminus \Sigma_d$.
Then set $z=E+i\eta$ with $\eta=\lambda^{2+\kappa}$, and
$f(x) = f_0(x/\ell)$ for $\ell = \lambda^c \lambda^{-1}\eta^{-1/2}$
for small $c$ to be specified momentarily.
Then by Corollary~\ref{cor:main-cor} combined with~\eqref{eq:rescaled-elliptic}, we have
\begin{equation}
\label{eq:weevil}
\begin{split}
|\lambda^{-2} \ell^{-2} \sum_x f(x/\ell) |R_{0x}|^2  - \tilde{u}(0)|
&\ldom \lambda^{-4}\ell^{-2}
(\eta^{-1} \ell^{-1} \|\nabla f_0\|_{C^0} + \|\Ft{f_0}\|_{L^1}) \\
&\qquad + \lambda^{-2}\ell^{-2} \eta^{-1}\Phi_d(E+i\eta),
\end{split}
\end{equation}
where $\tilde{u}$ solves
\[
(-\vartheta(E+i\eta) \Delta_{\Real^d} + \lambda^{2c} m(E+i\eta)) \tilde{u} = f_0.
\]
Note that $\tilde{u}$ depends on $\lambda$ implicitly. Plugging in the definition of $\eta$ and $\ell$ in~\eqref{eq:weevil} we obtain
\begin{equation}
\label{eq:f-scaling}
|\lambda^{-2} \ell^{-2} \sum_x f(x/\ell) |R_{0x}|^2  - \tilde{u}(0)|
\ldom  \lambda^{\kappa/2-3c} + \lambda^{-2c} \Phi_d(E+i\eta),
\end{equation}
which is bounded by $\lambda^{\delta'}$ for some $\delta'>0$, provided that $c$ is small enough.

Now $\tilde{u}$ can also be written as follows:
\begin{align*}
\tilde{u}
&= (-\vartheta(z) \Delta_{\Real^d} + \lambda^{2c} m(z))^{-1}f_0 \\
&= \vartheta(z)^{-1}(- \Delta_{\Real^d} + \lambda^{2c}\vartheta(z)^{-1} m(z))^{-1}f_0.
\end{align*}
In $d\geq 3$ we we have that $(-\Delta + \eps)^{-1}f_0 \to -\Delta^{-1} f_0$ pointwise as $\eps\to 0^+$.Therefore also
using~\eqref{eq:massdiff-limit} we have
\[
\lim_{\lambda \to 0} \tilde{u}(0) = c_d \beta_E \int f_0(y) |y|^{2-d}\diff y,
\]
where
\[
\beta_E := \frac{4d}{\pi}\rho^3(E) \nu(E)^{-1}
\]
and $c_d$ is such that\footnote{Alternatively, $c_d^{-1}$ is the surface area of the unit sphere.} $\Delta (c_d|y|^{2-d}) = \delta_0$.  The limit
\[
\lim_{\lambda\to 0} \lambda^{-2} \ell^{-2}
\sum_{x\in\bbZ^d}f(x/\ell) |R_{0x}|^2
= c_d \beta_E \int f_0(y) |y|^{2-d}\diff y
\]
now follows from a standard argument involving the Borel-Cantelli lemma and the continuity of the resolvent in $\lambda$.
In $d=2$ the \textit{gradient} of the Green's function converges pointwise, so the same result holds
provided $f_0=\Div \vec{v}$ and one replaces $|y|^{2-d}$ by $-\log(|y|)$.

\subsection{Proof of Theorem~\ref{thm:deloc}}
\label{sec:deloc-pf}
Theorem~\ref{thm:deloc} follows from Lemma~\ref{lem:radial-cutoff} and the following
deterministic lemma, which is based on ~\cite[Proposition 7.1]{Y3Teq}. Note that the following lemma proves a stronger statement than Theorem \ref{thm:deloc}
as it describes eigenfunctions localized to \textit{specific} regions of space,
as well as specific energies.

For any self-adjoint Hamiltonian $\in \mathcal{B}(\ell^2(\mathbb{Z}^d_L))$, we write
$(\psi_j^A,E_j^A)_{j\in[1,L^d]}$ to be an orthonormal basis of eigenfunctions for $H$ which satisfy $A\psi_j^A = E_j^A\psi_j^A$.
Further, for any $r>0$ and $x_0\in \mathbb{Z}^d_L$
we define the set
\[
\mathfrak{L}(A,r,x_0) := \{E_j^A \in \sigma(A) \mid
\|\psi_{j}^A \|_{\ell^2 (B_{r}(x_0))} \geq 1- r^{-4d}\}.
\]
Note that $\mathfrak{L}(A,r,x_0)$ contains eigenfunctions of
$A$ exponentially localized at scale $\ll r$ near $x_0$.

\begin{lemma}
\label{lem:deterministic-localization-length}
Let $E_0\in \R$, $\eta,\delta>0$ and $x_0\in \mathbb{Z}^d_L$. Suppose for some self-adjoint $A\in \mathcal{B}(\ell^2(\mathbb{Z}^d_L))$
\begin{align*}
\sum_{x\in B_r(x_0)}|G(E+i\eta)_{x_0,x}|^2 &\leq \delta \eta^{-1} \\
\sup_{\substack{|x-x_0|\leq r\\E\in [E_0-\eta,E_0+\eta]}}|G(E+i\eta)_{xx}| &:= D,
\end{align*}
where $G(z)=(A-z)^{-1}$ is the resolvent of $A$.
Then
\[
|\mathfrak{L}(A,r,x_0) \cap [E_0-\eta,E_0+\eta]| \lsim (\delta \eta r^d + r^{-2d}\delta^{-1} \eta) D.
\]
\end{lemma}

Note the phase space volume of the region $\{(x,\xi) \mid |x-x_0| \leq r, |\omega(\xi)-E|\leq \eta\}$ is about
$\eta r^d$ so one expects the upper bound
\[
|\mcal{L}(H,r,x_0)\cap [E_0-\eta,E_0+\eta]| \lsim \eta r^d.
\]
In fact, the condition $|G_{xx}(E+i\eta)| \leq D$ implies
$|\mathfrak{L}(A,r,x_0)\cap [E_0-\eta,E_0+\eta]| \lsim D \eta r^d$ since
\begin{equation}
\label{eq:basic-B-bound}
\begin{aligned}
|\mathfrak{L}(A,r,x_0)\cap [E_0-\eta,E_0+\eta]|
& \lsim \sum_{E_j^H\in \mathfrak{L}\cap[E_0-\eta,E_0+\eta]}\sum_{x\in B_{r}(x_0)}|\psi_j^H(x)|^2\\
& \lsim \eta \sum_{x\in B_r(x_0)}\Impt G_{xx}(E_0 + i\eta)\\
& \lsim D \eta r^{d}.
\end{aligned}
\end{equation}

\begin{proof}[Proof of Theorem \ref{thm:deloc} ]

Recall that $H_L=\Delta_L + \lambda V$ and let $\eta=\lambda^{2+\kappa_d-\delta}$, $r=\lambda^{-1}\eta^{-1/2}$, and $E_0\in[-2d,2d]$ with $d(E_0,\Sigma_d)>\eps$.  Furthermore, define
$$\mathfrak{L}(H_L,r):= \{E_j^{H_L}\ |\ \exists\  x_0\in \mathbb{Z}^d_L \text{ such that } \|\psi_j^H\|_{\ell^2(B_r(x_0))}\geq 1-r^{-4d}\}.$$

By Lemma~\ref{lem:radial-cutoff} and Proposition~\ref{prp:Zd-to-ZdL}, we have that
\[
\sum_{x\in B_r(x_0)}|R(E+i\eta)_{x_0,x}|^2
\ldom \lambda^{c\delta} \eta^{-1}
\]
for $\eta > \lambda^{2+\kappa_d-\delta}$ and $r <\lambda^\delta (\lambda^{-1}\eta^{-1/2})$.
Therefore we can apply Lemma~\ref{lem:deterministic-localization-length} to see that with  probability at least $1-\lambda^{N}$ for some large $N$,
\[
|\mcal{L}(H,x_0,r)\cap [E_0-\eps,E_0+\eps]| \lsim \lambda^{c\delta} \eps r^d
\]
for all $x_0\in\bbZ^d_L$.

Now letting $\{x_j\}$ form a $2r$-net of $\bbZ^d_L$, we have
\[
\mathfrak{L}(H,r) = \bigcup_j \mathfrak{L}(H,x_j,5r),
\]
and thus with probability at least $1-\lambda^{N}r/L$ we have
\[
|\mathfrak{L}(H,r)| \lsim \lambda^{c\delta}\eps L^d.
\]
On the other hand, by  Theorem~\ref{prp:zero-bd}, the bound
\[
|\sigma(H) \cap [E-\eps,E+\eps]| \gtrsim \eps L^d
\]
holds with probability $1-\lambda^N$ so long as $\lambda\ll \eps$. Taking $N$ sufficiently large now completes the proof.
\end{proof}

We now conclude with the proof of Lemma~\ref{lem:deterministic-localization-length}.

\begin{proof}[Proof of Lemma \ref{lem:deterministic-localization-length}]
We write $G$ for operator $G(E+i\eta)$, and define the set
$$\mathfrak{B}:= \mcal{L}(H,r,x_0)\cap[E_0-\eta,E_0+\eta],$$
and define the operator
\[
\One_{x_0,r} :=
\sum_{y\in B_r(x_0)}
\ket{y}\bra{y}.
\]
With this definition,
\[
\sum_{x\in B_r(x_0)}
|G(E+i\eta)_{x_0,x}|^2
= (G(E+i\eta) \One_{x_0,r}G(E+i\eta)^*)_{xx}.
\]
The idea is to estimate
$$\sum_{x\in B_{r}(x_0)} \|(1-\One_{x_0,r})G\delta_x\|^2 = \sum_{x\in B_r(x_0)} \braket{x|G(1-\One_{x_0,r})G^*|x},$$
in two ways.
First the Ward identity and our bound on $\braket{x|G \One_{x_0,r} G^*|x}$ implies
\begin{equation*}
\braket{x|G(1-\One_{x_0,r})G^*|x}
= \sum_{y}|G_{xy}|^2 -  \braket{x|G\One_{x_0,r}G^*|x}
\geq \frac{\Impt G_{xx}}{\eta}- \delta \eta^{-1}
\end{equation*}
so that
\begin{equation}
\label{eq:one-bd}
\sum_{x\in B_r(x_0)}
\braket{x|G(1-\One_{x_0,r})G^*|x} \geq \eta^{-1} \sum_{x\in B_r(x_0)} \Impt G_{xx} -  |B_r(x_0)| \delta\eta^{-1}.
\end{equation}

On the other hand letting $z:= E + i\eta$ we have for each $x$ the upper bound
\begin{equation}
\label{eq:two-bd}
\begin{aligned}
\|(1-\One_{x_0,r}) G\delta_x\|^2
&= \|(1-\One_{x_0,r}) \sum_j \frac{1}{E_j-z} \psi_j(x) \psi_j\|^2 \\
&\leq (1+\alpha) \|\sum_{j\not\in\mcal{B}} \frac{1}{E_j-z} \psi_j(x) \psi_j\|^2 \\
&\qquad + (1+\alpha^{-1})\|(1-\One_{x_0,r}) \sum_{j\in\mcal{B}} \frac{1}{E_j-z} \psi_j(x) \psi_j\|^2 \\
&=: S_1(x) + S_2(x),
\end{aligned}
\end{equation}
where $\alpha\in (0,1)$ is a small parameter to be chosen later. We estimate
$S_2$  using the definition of $\mathfrak{B}$:
\begin{align*}
\sum_{x\in B_r(x_0)} S_2(x)
&\leq (1+\alpha^{-1}) |\mathfrak{B}| \sum_{j\in\mcal{B}} \frac{1}{|E_j-z|^2} \sum_{x\in B_r(x_0)}|\psi_j(x)|^2
\|(1-\One_{x_0,r})\psi_j\|^2 \\
&\lsim \alpha^{-1} |\mathfrak{B}| \eta^{-2} r^{-4d}\sum_{j\in\mathfrak{B}}  \sum_{x\in B_r(x_0)}|\psi_j(x)|^2\\\\
& \lsim \alpha^{-1}r^{-3d}\eta^{-2} |\mathfrak{B}|\\
& \lsim D\alpha^{-1}r^{-2d}\eta^{-1} .
\end{align*}
We used \eqref{eq:basic-B-bound} to pass to the last line, and $\sum_{j\in \mathfrak{B}}|\psi_j(x)|^2 \leq 1$ for each $x$ to pass to the second to last line.
We can upper bound $S_1$ using orthogonality of eigenfunctions and the Ward identity to get
\begin{align*}
\sum_{x\in B_r(x_0)} S_1(x)
&\leq (1+\alpha) \sum_{x\in B_r(x_0)}
\| \sum_{j\not\in \mathfrak{B}} \frac{1}{E_j-z} \psi_j(x) \psi_j\|^2 \\
&\leq (1+\alpha) \sum_{x\in B_r(x_0)}
\sum_{j\not\in\mathfrak{B}} \frac{1}{|E_j-z|^2} |\psi_j(x)|^2 \\
&= (1+\alpha) \eta^{-1}\sum_{x\in B_r(x_0)} \Impt G_{xx}
- (1+\alpha) \sum_{j\in \mathfrak{B}} \frac{1}{|E_j-z|^2}\sum_{x\in B_r(x_0)} |\psi_j(x)|^2  \\
&\leq \eta^{-1} \sum_{x\in B_r(x_0)} \Impt G_{xx}
+ CD\alpha \eta^{-1} |B_r(x_0)|
- \frac{1+\alpha}{4}\eta^{-2} |\mathfrak{B}|.
\end{align*}

Combining the above two estimates with \eqref{eq:two-bd} gives
\begin{align*}
\sum_{x\in B_r(x_0)} \|(1-\One_{x_0,r})G\delta_x\|^2
\leq \eta^{-1}\sum_{x\in B_r(x_0)}\Impt G_{xx} - \frac{1+\alpha}{4}\eta^{-2}|\mathcal{B}| + CD\alpha^{-1}\eta^{-1}r^{-2d} + CD \alpha \eta^{-1}r^d.
\end{align*}
Subtracting this from \eqref{eq:one-bd} and rearranging gives
\begin{align*}
\frac{1+\alpha}{4}\eta^{-2}|\mathfrak{B}|\lsim D\alpha^{-1}\eta^{-1}r^{-2d} + (\alpha+\delta) D\eta^{-1}r^d
\end{align*}
Taking $\alpha = \delta$ proves the result.
\end{proof}

\appendix

\section{Elementary estimates for the free resolvent}

In this section we collect some elementary estimates about the resolvent of
the Laplacian on the lattice $(\Delta_{\bbZ^d}-z)^{-1}$ and on the torus $(\Delta_{\bbZ^d_L}-z)^{-1}$
which are used throughout the paper.  Throughout this appendix we write $z=E+i\eta$, and for any $\epsilon>0$ define the domain
$$D_{\epsilon} := \{E+i\eta\ |\ d(E,\Sigma_d)\geq\epsilon,\ \eta>0 \}.$$

\subsection{Comparing $\Delta_{\bbZ^d}$ and $\Delta_{\bbZ^d_L}$}

The Combes-Thomas estimate Theorem 10.5 of~\cite{aizenman2015random} applied to the free Schrodinger operator implies the following exponential decay bound.
\begin{lemma}
\label{lem:exp-decay}
For $|x-y|>C\eta^{-1}$, the resolvent entries satisfy
\[
|\braket{x|(\Delta_{\bbZ^d} - z)^{-1}|y}| \leq C \exp(-c \eta |x-y|).
\]
\end{lemma}

To transfer bounds for the resolvent on $\bbZ^d$ to the resolvent on $\bbZ^d_L$ we use the identity
\[
\braket{x|(\Delta_{\bbZ^d_L}-z)^{-1}|y}
= \sum_{k\in\bbZ^d} \braket{x|(\Delta_{\bbZ^d}-z)^{-1}|(y+kL)}.
\]
Combined with the Combes-Thomas estimate above this implies the following result.
\begin{lemma}
\label{cor:resolvent-truncation}
Suppose $z= E+i\eta$, $\eta<1$ and $L \gg \eta^{-2}$.
Then for $x,y\in\bbZ^d$ such that $|x-y+Lk| > |x-y|$ for all $k\in\bbZ^d\setminus\{0\}$, we have
\[
|\braket{x|(\Delta_{\bbZ^d_L}-z)^{-1}|y} -
\braket{x|(\Delta_{\bbZ^d}-z)^{-1}|y}| \leq \exp(-c \eta L).
\]
In particular it follows that for all $1\leq p\leq \infty$,
\[
C_p^{-1}\|(\Delta_{\bbZ^d_L}-z)^{-1}\ket{0}\|_{\ell^p(\bbZ^d_L)}\leq \|(\Delta_{\bbZ^d}-z)^{-1}\ket{0}\|_{\ell^p(\bbZ^d)}
\leq C_p
\|(\Delta_{\bbZ^d_L}-z)^{-1}\delta_0\|_{\ell^p(\bbZ^d_L)}.
\]
\end{lemma}

\subsection{Density of states estimate}

Next we state a bound for the density of states for the Laplacian on $\bbZ^d$.  To state the bound we
introduce the measure $\gamma$ on $\Real$ which is the pushforward of the (normalized)
Lebesgue measure on $\Torus^d$
by the dispersion relation $\omega:\Torus^d \to\Real$.  That is, $\gamma$ satisfies
\[
\int_{\mathbb{R}} f(\sigma) \diff\gamma= (2\pi)^{-d} \int_{\Torus^d} f(\omega(p)) \diff p.
\]

This density is actually smooth outside of a finite set of points, and moreover is bounded in $d\geq 3$.
\begin{lemma}
\label{lem:free-dos}
The measure $\gamma$ has a density $\gamma = \rho(x)\diff x$ which is supported on $[-2d,2d]$ and
satisfies for $d\geq 2$ the upper bound
\[
|\rho(\sigma)| \leq C (1 + 1_{d=2}\log(|\sigma|^{-1})
\]
and a lower bound on $[-2d+\eps,2d-\eps]$:
\[
\rho(\sigma) \geq c_\eps.
\]
Moreover, $\rho$ is smooth outside of $\Sigma_d$.
\end{lemma}

\begin{proof}
By the coarea formula the density $\rho(\sigma)$ is given by
\[
\rho(\sigma) = (2\pi)^{-d} \int_{X_\sigma} \frac{\diff \mcal{H}^{d-1}(p)}{|\nabla \omega(p)|},
\]
where $X_\sigma = \{p\in\Torus^d \mid \omega(p)=\sigma\}$ and $\diff\mcal{H}^{d-1}$ is the $d-1$-dimensional
Hausdorff measure.  It is already immediate from this formula that $\rho$ is smooth outside of $\Sigma_d$,
which is the set of critical values of $\omega$.

Let $\mcal{C}_d\subset\Torus^d$ be the set of critical points of $\omega$, which are
all nondegenerate.  Then $\nabla \omega(p)$ satisfies
\[
|\nabla \omega(p)| \geq c \min_{q\in\mcal{C}_d} |p-q|.
\]
Therefore
\[
\frac{1}{|\nabla \omega(p)|} \leq C \sum_{q\in\mcal{C}_d} |p-q|^{-1}.
\]
Each of the level surfaces $X_\sigma$ satisfies the ball growth condition
\[
\mcal{H}^{d-1}(B_r \cap X_\sigma) \leq Cr^{d-1},
\]
so the singularity of $|\nabla \omega(p)|^{-1}$ is integrable on $X_\sigma$ when $d\geq 3$, meaning we have
\begin{equation}
\label{eq:vh-good}
\rho(\sigma) = \int_{X_\sigma} \frac{dp}{|\nabla e(p)|} \leq C.
\end{equation}
This implies the upper bound for $d\geq 3$.

In dimension $d=2$ there is a logarithmic singularity at $\sigma=0$ (owing to the fact that $|x|^{-1}$ is barely not integrable in one dimension), so instead we replace~\eqref{eq:vh-good} by
\[
\int_{X_\sigma} \frac{dp}{|\nabla \omega(p)|} \leq C \log |\sigma|^{-1}.
\]

For the lower bound, simply notice that $|\nabla \omega(p)| \leq 2d$, so $\rho(\sigma) \geq |X_\sigma|$.
The function $\sigma\mapsto |X_\sigma|$ is positive on $(-2d,2d)$ and continuous so the lower bound follows
from compactness.
\end{proof}

\subsection{Regularity of the resolvent entries}

Lemma~\ref{lem:free-dos} implies some regularity for the entries of the resolvent which we use
to analyze the stability of the self-consistent equation. We have regularity on
the domains $D_\eps$, given by
$$D_{\eps} := \{E+i\eta\ |\ d(E,\Sigma_d)\geq\eps,\ \eta>0 \}.$$
\begin{lemma}
\label{lem:dos-regularity}
The function $\phi_L(z) := (\Delta_{\bbZ^d_L}-z)^{-1}_{00}$ satisfies
\[
\|\phi_L\|_{C^2(D_\eps)} \leq C,
\]
where $C$ depends implicitly on $\eps$ but not on $\lambda$ or $L$.
Moreover,
\begin{align}\label{eq:diagonalResolventLowerBound}
	\Impt\phi_L\gtrsim 1
\end{align}
\end{lemma}
\begin{proof}
Because $L\gg \lambda^{-10}$ it suffices to work with $\phi(z) := (\Delta_{\bbZ^d}-z)^{-1}_{00}$, which can be written as
\[
\phi(z) = \int_{\Torus^d} \frac{dp}{\omega(p)-z}
= \int \frac{1}{\sigma-z} \rho(\sigma)\diff \sigma.
\]
This is the analytic continuation to the upper half plane of the function $H\rho + i\rho$, where
$H\rho$ is the Hilbert transform of $\rho$.  Since $\rho$ is smooth away from $\Sigma_d$,
the Hilbert transform $H\rho$ is as well.
\end{proof}

A consequence of Lemma~\ref{lem:dos-regularity} and the Ward identity~\eqref{eq:Ward} is the following estimate for the resolvent.
\begin{corollary}
\label{cor:free-onetwo}
For $z=E+i\eta$ with $d(E,\Sigma_d)>\eps$ and $\eta>0$,
\[
\|(\Delta_{\bbZ^d} - z)^{-1}\|_{1\to 2} \lsim \eta^{-1/2}.
\]
\end{corollary}

\subsection{Uniqueness of the self-consistent equation}
Finally we record facts about the equation
\[
\theta = ((\Delta_L - (z+\lambda^2\theta))^{-1})_{00}.
\]
\begin{proposition}
\label{pr:thetaAPriori}
For any $z\in D_{\epsilon}$, there is a unique $\theta(z)\in \bbH$ solving
\begin{align*}
\theta(z)=(\Delta_L-(z+\lambda^2\theta(z))^{-1}_{00}
\end{align*}
Moreover for $\lambda \ll 1$, $\theta(z)$ satisfies
\[
|\Impt(\theta(z)) - \rho(z)| \lsim \lambda^2.
\]
\end{proposition}

\begin{proof}
Fix $z\in D_{\epsilon}$. First we show existence and uniqueness.
Define the map $\Phi_z:\mathbb{H}\to\mathbb{H}$ by
$$\Phi_z(w) := (\Delta_L - (z+\lambda^2 w))^{-1}_{00} =  \frac{1}{L^d}\sum_{\xi\in (\frac{2\pi}{L}\bbZ)^d/2\pi \bbZ^d} \frac{1}{\sum_{i=1}^d 2\cos(\xi_i) - z- \lambda^2w}.$$
For any $r\in \R$, the mapping $w\mapsto \frac{1}{r-z-\lambda^2 w},$ maps $\mathbb{H}$ to a compact subset of $\mathbb{H}$, and is a strict contraction in the standard hyperbolic distance on $\mathbb{H}$, since $\Impt z>0$.
$\Phi_z$ is an average of such maps, so $\Phi_z$ has a unique fixed point in the upper half plane.

Now we show $|\Impt(\theta(z)) - \rho(z)| \lsim \lambda^{2}$. Note that it suffices to show
$|\theta(z)|\lsim 1$, since then for $\lambda$ small enough we could invoke Lemma \ref{lem:dos-regularity} to conclude that
$$|\theta(z) - (\Delta_L - z)_{00}^{-1}|\leq |(\Delta -(z+\lambda^2\theta(z)))^{-1}_{00}-(\Delta_L-z)^{-1}_{00}|\lsim \lambda^2.$$
To prove $|\theta(z)|\lsim 1$, we consider the action of $\Phi_z$ on the set
$B_{.5\epsilon\lambda^{-2}}(0)\cap \{w\ | \Impt w\geq\lambda\}$. Since $\lambda \ll 1$ Lemma \ref{lem:dos-regularity} implies that for any $w\in B_{.5\epsilon\lambda^{-2}}(0)\cap \{w\ | \Impt w\geq\lambda\}$
$$|\Phi_z(w)|\lsim 1,$$
and
$$\Impt \Phi_z(w)\gtrsim 1.$$
The Brouwer fixed point theorem implies $\Phi_z$ has a fixed point in the set $B_{.5\epsilon\lambda^{-2}}(0)\cap \{w\ | \Impt w\geq\lambda\}$ and so by the uniqueness proven above this fixed point must be $\theta(z)$. Then the bounds above imply $|\theta(z)|\lsim 1.$

\end{proof}

\printbibliography

\end{document}